\documentclass[reprint,prx,amsmath,longbibliography,floatfix]{revtex4-2}
\usepackage[utf8]{inputenc}
\usepackage[T1]{fontenc}
\usepackage{bbold}
\usepackage{txfontsb}
\usepackage{times}

\usepackage{amsthm,mathtools,amsfonts}
\usepackage{calc} 
\usepackage{graphicx}
\usepackage{tabularx}
\usepackage[citecolor=blue, urlcolor=blue, linkcolor=blue, colorlinks=true]{hyperref}

\graphicspath{{figures/}{disorder average/}}

\newcommand*\dagg{^{\dagger}}

\newcommand{\D}{\mathcal{D}}
\renewcommand{\P}{\mathcal{P}}
\newcommand{\R}{\mathcal{R}}

\newcommand{\OO}{\mathcal{O}}
\newcommand{\eps}{\varepsilon}
\newcommand{\Eps}{\mathcal E}

\newcommand*{\ket}[1]{|#1\rangle}
\newcommand*{\bra}[1]{\langle#1|}

\newcommand*{\mat}[1]{\begin{pmatrix}#1\end{pmatrix}}
\newcommand*{\floor}[1]{\lfloor#1\rfloor}
\newcommand{\ce}{=}
\newcommand{\Nto}{\xrightarrow[N\to\infty]{}}
\newcommand{\Pcont}{\bar P}
\newcommand{\PPcont}{\mathcal P^{\mathrm{cont}}}
\newcommand{\Pex}{P^*}
\newcommand{\PPex}{\mathcal P^*}
\newcommand{\PPap}{\tilde{\mathcal P}}
\newcommand{\rhonr}{\mathbb P}
\newcommand{\tildeN}{N-2r}
\newcommand{\Q}{\mathcal Q}

\usepackage[capitalize]{cleveref}
\crefformat{section}{Sec.~#2#1#3}

\newtheorem{theorem}{Theorem}
\newtheorem{lemma}{Lemma}
\newtheorem{corollary}{Corollary}
\newtheorem*{L2}{Corollary~\ref{lm:mag}}
\newtheorem*{T1}{Theorem~\ref{th:pure}}
\newtheorem*{T2}{Theorem~\ref{th:loss}}
\newtheorem{definition}{Definition}

\begin{document}
\title{Large-$N$ limit of Dicke superradiance}
\author{Daniel Malz}
\author{Rahul Trivedi}
\author{Ignacio Cirac}
\affiliation{Max Planck Institute for Quantum Optics, Hans-Kopfermann-Straße 1, D-85748 Garching, Germany}
\affiliation{Munich Center for Quantum Science and Technology, Schellingstraße 4, D-80799 München, Germany}

\begin{abstract}
  We investigate the thermodynamic limit of Dicke superradiance.
  We find an expression for the system's density matrix that we can prove is exact in the limit of large atom numbers $N$. This is in contrast to previously known solutions whose accuracy has only been established numerically and that are valid only for a range of times.
  We also introduce an asymptotically exact solution when the system is subject to additional incoherent decay of excitations as this is a common occurrence in experiment. 
\end{abstract}
\maketitle

\section{Introduction}
Superradiance was introduced by Dicke in 1954~\cite{Dicke1954}. 
In a simple model of all-to-all coupled emitters, he showed that during their decay process, quantum coherence is established spontaneously and leads to a superradiant burst of photons in which the maximum decay rate scales as the square of the number of atoms $\Gamma_{\mathrm{max}}\sim \gamma_0N^2/4$, where $\gamma_0$ is the decay rate of an individual atom.
Superradiance can be construed as synchronization of many emitters~\cite{Zhu2015}, akin to the onset of lasing. Moreover, since quantum effects are typically difficult to observe in large ensembles of atoms, Dicke's observation stirred a lot of interest in the 70s~\cite{Bonifacio1971,Bonifacio1971a,Degiorgio1971a,Degiorgio1971,Haake1972,Agarwal1974,Narducci1974a}, as reviewed by Gross and Haroche~\cite{Gross1982}.
Superradiance was observed experimentally, also in the 70s~\cite{Skribanowitz1973,Gross1976,Vrehen1977,Raimond1982}.
A closely related, but different, phenomenon is the Dicke phase transition~\cite{Hepp1973,Wang1973}, which we will not consider here.

Recent experimental advances to produce arrays of quantum emitters in optical lattices~\cite{Bakr2010,Sherson2010,Greif2016,Kumar2018} have led to a revival in interest in their radiation properties~\cite{Porras2008,Svidzinsky2010,Jenkins2016,Bettles2016,Asenjo-Garcia2017,Chang2018,Lemberger2021,Masson2020,Masson2021,Sierra2021,Glicenstein2021,Gold2021,Ferioli2021}.
This is also motivated by technological applications of superradiance (and subradiance) as a way to improve light--matter interfaces, which play a central role in many quantum technology platforms~\cite{Fleischhauer2005,Hammerer2010,Peyronel2012}.
Several experiments have probed radiance in the linear regime~\cite{Rohlsberger2010,Roof2016,Guerin2016,Araujo2016,Kim2018,Pennetta2021,Pennetta2021a}.
Another possibility is raised by the emergence of waveguide QED~\cite{LeKien2005,LeKien2008,Goban2014,Turschmann2019},
where quantum emitters are coupled to one-dimensional optical fields, with diverse implementations ranging from cold atoms near waveguides~\cite{Vetsch2010,Goban2012,Thompson2013,Goban2015,Gouraud2015}, quantum dots~\cite{Akimov2007,Lodahl2015}, or nitrogen-vacancy centres~\cite{Huck2011,Sipahigil2016,Evans2018}.
Strong confinement of light in the waveguide allows for substantial cooperativities (ratio of waveguide to free-space decay rates).
In the absence of propagation losses, the one-dimensional nature of the light field gives rise to effectively infinite-range dissipative interactions. Moreover, if the atoms are spaced by the wavelength of their dipole transition $\lambda$ (the ``atomic-mirror configuration'') their dynamics is described by the original Dicke model.
Superradiance has also been explored for technological applications such low-linewidth lasers~\cite{Meiser2009,Bohnet2012}, metrology~\cite{Paulisch2019}, or sensing~\cite{Yang2021}.

In light of these developments, we re-visit the theory of superradiance.
Approximate solutions have been derived in a number of different ways~\cite{Degiorgio1971,Degiorgio1971a,Haake1972,Agarwal1974,Narducci1974a,Gross1982},
but ultimately they all rely on a continuum limit in which the magnetization can take non-integer values.
This elegant description is, however, incorrect at short and long times. Additionally, it sensitively depends on a chosen initial distribution, and thus needs a rigorous justification.
Since Dicke superradiance is a fundamental model of quantum optics, it is remarkable that no solution exact in the limit $N\to\infty$ has been derived so far.

The main contribution of the present work is an explicit formula for the density matrix as a function of time, alongside a proof that it is asymptotically exact as $N$ becomes large. In contrast to previous work, our solution works for all times.
In a second contribution, motivated by experimental realities, we derive the solution in the presence of incoherent loss (for example through decay into free space rather than guided modes) and show that it is exact as $N\to\infty$.

\section{Setup}
We consider the quantum master equation of $N$ two-level systems (``atoms'' or ``emitters'') with states $\ket0$, $\ket1$ that are subject to both collective and incoherent decay,
\begin{equation}
  \dot\rho(t)=\D[S^-]\rho(t) + \gamma \sum_{n=1}^N\D[\sigma_n^-]\rho(t),
  \label{eq:QME}
\end{equation}
where $\sigma_n^-=\ket{0}_n\bra1$, $S^-=\sum_m\sigma_m^-$, and $\D[a]\rho=a\rho a\dagg-(1/2)(a\dagg a\rho + \rho a\dagg a)$.
Note that we have omitted the transition frequency of the atoms from the description, as it does not affect they decay dynamics.
In \cref{eq:QME}, the dimensionless parameter $\gamma$ controls the ratio of incoherent to collective decay.
For $\gamma=0$, the model originally studied by Dicke is recovered (to be distinguished from the ``Dicke model'' introduced by Hepp and Lieb~\cite{Hepp1973}).

Any permutation of the atoms is a symmetry of \cref{eq:QME}.
Thus, a permutation-invariant initial state remains so, i.e., $\rho(t)=\Pi_s[\rho(t)]$, where $\Pi$ is the permutation operator and $s$ is any permutation of the atoms.
Defining the $z$-component of the collective angular momentum $S^z=\sum_n\sigma_n^z$ and $S^2=S^+S^-+(S^z)^2$, we can label collective states of $N$ spin-$1/2$ particles according to their total angular momentum $j$ and its $z$-component $m$.
The Hilbert space of $N$ two-level systems is then spanned by states $\{\ket{j,m,\alpha}\}$, which obey $S^2\ket{j,m,\alpha}=j(j+1)\ket{j,m,\alpha}$ and $S^z\ket{j,m,\alpha}=m\ket{j,m,\alpha}$, where
$0\leq\alpha< d_j$ is an integer accounting for their multiplicity $d_j=N!(2j+1)/((N/2-j)!(N/2+j+1)!)$.

We would like to solve \cref{eq:QME} when the system is initialized in the (permutation-invariant) state corresponding to fully excited atoms $\rho(0)=\ket{1\cdots1}\bra{1\cdots1}$.

\subsection{Dicke model}
The original Dicke model is obtained by setting $\gamma =0$~\cite{Dicke1954}, and we call this situation ``pure superradiance''.
In this case, the system explores only the Dicke states with maximum total spin $\ket n\ce\ket{N/2,-N/2+n}$, where $n\in\{0,1,\cdots, N\}$.
The density matrix of the system is diagonal at all times, $\rho(\tau) =\sum_n P_n(\tau)\ket{n}\bra{n}$, and thus \cref{eq:QME} reduces to a rate equation
\begin{equation}
  \partial_\tau P_n(\tau) = -\gamma_nP_n(\tau)+\gamma_{n+1}P_{n+1}(\tau).
  \label{eq:p_n}
\end{equation}
Here, for convenience we rescaled time $\tau\ce Nt$, and defined
\begin{equation}
  \gamma_n=n(N-n+1)/N.
  \label{eq:pure_rates}
\end{equation}
The most salient feature of this model is the superradiant burst that occurs at $\tau\approx\ln N$ during which most excitations are emitted at a rate that scales with $N^2$.

\subsection{Dicke model with incoherent decay}
Free-space decay competes with superradiance, as it reduces the coherence of the ensemble and allows the system to explore states with lower total spin.
For convenience, we label these states by their number of ``dark'' excitations $r=N/2-j$, which measures by how much the total spin has been reduced, and the number of collective excitations $n=N/2-2r+m$ that can be removed from the state before it reaches the bottom of the Dicke ladder.
In terms of $n$ and $r$, the relevant states are projectors into the space of defined $n$ and $r$,
\begin{equation}
  \rhonr_{n,r} = \frac{1}{d_r}\sum_{\alpha=0}^{d_r-1}\left|\frac N2-r,r+n-\frac N2, \alpha\right\rangle\left\langle\frac N2-r,r+n-\frac N2, \alpha\right|,
  \label{eq:rho_nr}
\end{equation}
where in terms of $r$ the multiplicity reads $d_r=N!(N-2r+1)/(r!(N-r+1)!)$.
As before, \cref{eq:QME} reduces to a rate equation transitioning between these states, and the solution is a mixture $\rho(\tau)=\sum_{n,r}\P_{n,r}(\tau)\rhonr_{n,r}$~\footnote{We take the sum $\sum_{n,r}$ to run over all allowed values of $n$ and $r$, which are $r\in\{0\cdots\floor{N/2}\}$ and $n\in\{0,\cdots,N-2r\}$.},
  where the probabilities $\P_{n,r}$ obey
\begin{equation}
  \partial_\tau \P_{n,r} = -\Gamma^{(1)}_{n,r}\P_{n,r}+\Gamma^{(2)}_{n+1,r}\P_{n+1,r}+\Gamma^{(3)}_{n+2,r-1}\P_{n+2,r-1}+\Gamma^{(4)}_{n,r+1}\P_{n,r+1},
  \label{eq:pnr}
\end{equation}
and the rates are given through~\cite{Shammah2018}
\begin{subequations}
  \begin{align}
  	\Gamma^{(1)}_{n,r} &= \frac{n}{N}(\tildeN -n+1)+\frac{\gamma }{N}(n+r),\\
  	\Gamma^{(2)}_{n,r} &= \frac{n}{N}(\tildeN -n+1)+\frac{\gamma }{N}\frac{n(N+2)(\tildeN -n+1)}{(N-2r) (\tildeN +2)},\\
  	\Gamma^{(3)}_{n,r} &= \frac{\gamma }{N}\frac{n(n-1)(\tildeN +r+1)}{(N-2r)(\tildeN +1)},\\
  	\Gamma^{(4)}_{n,r} &= \frac{\gamma }{N}\frac{(\tildeN -n+1)(\tildeN -n+2)r}{(\tildeN +2)(\tildeN +1)}.
  \end{align}
  \label{eq:rates_main}
\end{subequations}
Here and in the following, we distinguish distributions in $n,r$ by using calligraphic font.

Note that collective decay (terms without $\gamma$) cannot change the total spin $j$ (and neither $r$), and its rate is reduced by the presence of dark excitations $r$.
All other terms, proportional to $\gamma$, are due to incoherent loss, which can either reduce $r$ ($\Gamma^{(4)}$), increase $r$ ($\Gamma^{(3)}$) or leave it unchanged ($\Gamma^{(2)}$).

\subsection{Previous work}
Several approaches have been used to describe pure Dicke superradiance. 
In principle, one can write down an exact iterative solution~\cite{Agarwal1970}, but it cannot be summed up in closed form to yield a formula for the magnetization for general $N$.
Approximate solutions for $P_n(t)$ were derived in different ways by Degiorgio~\cite{Degiorgio1971a}, Degiorgio and Ghielmetti~\cite{Degiorgio1971}, and Haake and Glauber~\cite{Haake1972}. All these solutions are equivalent and read 
\begin{equation}
  \Pcont_n(\tau)=\frac{N^2}{n^2}\exp\left[ -\tau-e^{-\tau}N(N-n+\lambda)/n \right].
  \label{eq:Pcont}
\end{equation}
with $\lambda=1$.
Gross and Haroche later provided another derivation, but quote $\Pcont$ with $\lambda=0$~\cite{Gross1982}.
The choice of $\lambda$ has little effect and in the following we take $\lambda=0$.
All approaches rely on the continuum limit of \cref{eq:p_n} (dropping the 1 in $\gamma_n$ to be consistent with the choice $\lambda=0$),
\begin{equation}
  \partial_\tau \Pcont_n(\tau)=\partial_n[n(N-n)\Pcont_n(\tau)]/N,
  \label{eq:pn_continuum}
\end{equation}
where $n$ is interpreted as a continuous variable $n\in[0,N]$.
From \cref{eq:Pcont}, Degiorgio calculated the average magnetization $\mu(\tau)=1/2+\langle S_z\rangle/N $ and radiance $\rho(\tau)=\langle S^+S^-\rangle/N^2 $ as functions of time
\begin{subequations}
  \begin{align}
  	\label{eq:mu_and_rho}
  	\mu(\tau) &= z e^z H(z), \quad
  	\rho(\tau) = z(\tau)-(1+z(\tau))\mu(z(\tau)),\\
  	H(z)&\ce\int_{z}^\infty[y e^y]^{-1}dy,\qquad z\ce Ne^{-\tau},
  \end{align}
\end{subequations}
which predict the maximum average radiance $\max_\tau\langle S^+S^-\rangle = 0.196 N^2/4$ at $\tau_{\mathrm{max}}=\ln N+0.330$~\footnote{The actual constant quoted in Ref.~\cite{Degiorgio1971}, 0.357, is slightly incorrect, but the difference has hardly any effect.}.
To arrive at \cref{eq:Pcont}, an initial distribution was chosen by requiring that the initial radiance distribution follow a Bose-Einstein distribution~\cite{Degiorgio1971a}.
Other authors arrive at the same conclusion by matching it to the corresponding Wigner distribution~\cite{Haake1972}, or to predictions from the master equation~\cite{Degiorgio1971,Lemberger2021}.
However, the approximation through the continuum equation \eqref{eq:pn_continuum} is not valid at short times and its results have not rigorously been shown to be correct~\footnote{This issue was commented on by Degiorgio and Ghielmetti~\cite{Degiorgio1971}, but not pursued further.}. As a result, the distribution \cref{eq:Pcont} is incorrect at short times. While the magnetization \cref{eq:mu_and_rho} agrees well with numerics at finite $N$, to our knowledge it has not rigorously been shown that it is valid in the limit $N\to\infty$, not even at sufficiently long times.

The case of $\gamma \neq0$ has received considerably less attention. The competition between collective and incoherent decay has been studied in some detail~\cite{Lee1976}, and a stochastic unravelling was used to efficiently simulate the system~\cite{Clemens2002}, but a complete theory is missing.
We note that in the context of the Dicke phase transition in cavity QED, a number of works have found a non-trivial interplay between dephasing, incoherent loss, and collective coupling~\cite{DallaTorre2013,DallaTorre2016,Kirton2017}.

\subsection{Our contribution}
Our first contribution is to derive a solution for pure superradiance that we rigorously prove to be correct in the large-$N$ limit for all times (\cref{th:pure}).
This is in contrast to the previously established solution, \cref{eq:Pcont}, which, as we show, is asymptotically exact only for intermediate times.
Specifically, we show that the one-norm of the difference between the exact probability vector $\Pex$ and the literature solution $\Pcont$ obeys
\begin{equation}
  \lim_{N\to\infty}| |\Pcont(\tau)-\Pex(\tau)| |_1\neq 0,\qquad \tau\text{ const}
  \label{eq:no_good_for_constant_times}
\end{equation}
but
\begin{equation}
  \lim_{N\to\infty}| |\Pcont(\alpha \ln N)-\Pex(\alpha\ln N)| |_1=0,\quad 0<\alpha<2.
  \label{eq:good_for_log_times}
\end{equation}
For $\alpha>2$, $| |\Pcont| |_1<N^{\alpha-2}\pi^2/6\to0$ as $N\to\infty$, which implies $| |\Pcont-P^*| |\to1$.
\Cref{eq:no_good_for_constant_times} follows from the fact that for constant times, the normalization of $\Pcont$ tends to $| |\Pcont| |_1\to \exp(-\tau-\lambda e^{-\tau})/(1-\exp(-e^{-\tau}))\neq1$ (\cref{lm:Pcont_norm}).
Since the exact solution $\Pex$ is normalized, the reverse triangle identity gives \cref{eq:no_good_for_constant_times}.

Our second contribution is to extend our analysis to $\gamma\neq0$, where we again provide a solution that is valid for all times (\cref{th:loss}).

\section{Large-$N$ limit of pure superradiance}
Our goal is to find an asymptotically exact solution $R_n$
in the sense that the one-norm of the difference to the exact solution $\Pex$ vanishes as $N\to\infty$, like in \cref{eq:good_for_log_times}, but we require that the solution works for all times.
We can prove that this is fulfilled by $R$ given through
\begin{equation}
  R_n(\tau) =
  \begin{cases}
  	R_n^<(\tau), &\tau\leq\tau_1\\
  	R_n^>(\tau), &\tau>\tau_1,
  \end{cases}
  \label{eq:Rn}
\end{equation}
where $\tau_1 = (1+\delta_1)\ln N$ with $\delta_1=2/5$, and
\begin{subequations}
  \begin{align}
  	\label{eq:R<}
  	R_n^<(\tau)&=\left(\frac{N}{n}\right)^2e^{-\tau}(1-e^{-\tau})^{N(N/n-1)},\\
  	R_n^>(\tau)&=\sum_{m=n}^N\mat{m\\n}e^{-m(\tau-\tau_1)}(e^{\tau-\tau_1}-1)^{m-n}R_m^<(\tau_1).
  	\label{eq:R>}
  \end{align}
  \label{eq:Rs}
\end{subequations}
There is a slight subtlety, since the time of the superradiant burst occurs depends on $N$, and thus it is not sufficient to prove that the solution converges at a constant time. Instead we allow for general sequences of times that may depend on $N$.
\begin{theorem}[Superradiant decay from all-inverted state]\label{th:pure}
  For any sequence of times $\{\tau_n > 0: n \in \mathbb{N}\}$, 
  \begin{equation}
  	\lim_{N\to \infty} ||\Pex(\tau_N) - R(\tau_N) ||_1 = 0.
  	\label{eq:convergence}
  \end{equation}
\end{theorem}

We only sketch the proof here, the full version can be found in \cref{app:proof}.
We separately consider short times (up to $\tau_0=\delta_0\ln N$, with $0<\delta_0<1$), intermediate times (between $\tau_0$ and $\tau_1=(1+\delta_1)\ln N$, with $0<\delta_1<1$), and late times ($\tau>\tau_1$).

At short times, we replace $\gamma_n\to\gamma^0_n=N-n+1$ in \cref{eq:p_n}, which is valid since the probability distribution is expected to have support only for $N-n\ll N$. The resulting equation is solved by 
\begin{equation}
  Q_n(\tau) = e^{-\tau}\left( 1-e^{-\tau} \right)^{N-n},\quad Q_n(0)\ce \delta_{n,N}.
  \label{eq:Q}
\end{equation}
Using this solution, we show that the error due to the substitution $\gamma_n\to\gamma_n^0$ vanishes as $N\to\infty$ as long as we only consider times up to $\tau_0=\delta_0\ln N$ with $0<\delta_0<1$.
We then use $Q(\tau_0)$ as the initial condition for \cref{eq:pn_continuum} and thus obtain a solution up to $\tau_1$. Since $Q_n(\tau_0)$ is sufficiently smooth in $n$, we can show that the error incurred by using the continuum approximation vanishes as $N\to\infty$. A further approximation of the result yields $R^<$ [\cref{eq:R<}].
For late times, we again linearize the equations of motion, $\gamma_n\to\gamma_n^1\ce n$, which is valid as the probability distribution is expected to concentrate around $n\approx0$. This yields the solution $R^>$ in \cref{eq:R>}. We can bound the error to show that $| |R^>-\Pex| |_1\to0$.

A few comments are in order about the solution we provide. 
First, for any finite $N$, $R$ neither exactly fulfils the equation of motion \cref{eq:p_n}, nor the continuum equation \cref{eq:pn_continuum}, and its derivative is discontinuous at $\tau_1$, although $R$ itself is continuous.
Nevertheless, $R$ describes the correct distribution at all times.
Second, as part of our proof, we show that $| |\Pex(\tau)-\Pcont(\tau)| |_1\to0$ as $N\to\infty$, which proves that the solution found in the literature also converges to the right solution, but only for intermediate times.
Third, a remaining question is whether the magnetization and radiance~\eqref{eq:mu_and_rho} are accurate, since they are calculated by integrating the continuum distribution rather than summing the (true) discrete distribution. With the following corollary, we show that both formulae are nevertheless asymptotically exact for all times.

\begin{corollary}[Magnetization and radiance]\label{lm:mag}
  The magnetization and radiance~\eqref{eq:mu_and_rho} predicted in the continuum limit are asymptotically correct for any sequence of times 
  $\{\tau_n > 0: n \in \mathbb{N}\}$
  \begin{subequations}
  	\begin{align}
  	  \lim_{N\to\infty}\mu^{\mathrm{exact}}(\tau_N)-\mu(\tau_N)&\to0,\\
  	  \lim_{N\to\infty}\rho^{\mathrm{exact}}(\tau_N)-\rho(\tau_N)&\to0.
  	\end{align}
  \end{subequations}
\end{corollary}

\begin{figure}[t!]
  \centering
  \includegraphics[width=\linewidth]{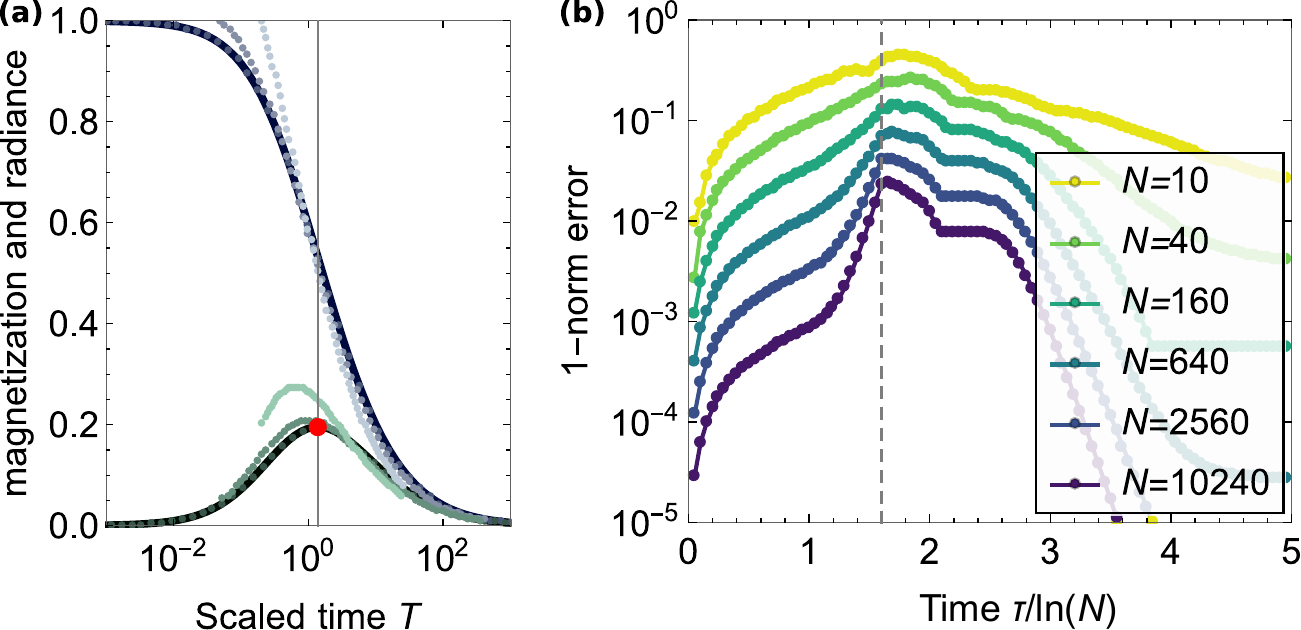}
  \caption{
  	(a) Magnetization $\mu$ (upper blue solid line) and radiance $\rho$ (lower green solid line) as defined in \cref{eq:scaling_form}, together with numerical data for $N=5,20,1000$ (light to dark) showing quick convergence.
  	The largest radiance is $\rho\approx0.196$ at $T\approx1.391$, marked with a red dot.
  	(b) Numerical evaluation of the 1-norm error \cref{eq:convergence}, which vanishes as $N\to\infty$ by \cref{th:pure}.
  	The dashed line marks $\tau_1=1.6\ln N$.
  }
  \label{fig:convergence}
\end{figure}

One might wonder whether a universal form for the magnetization and radiance can be found that is independent of $N$ and thus holds in the large-$N$ limit.
We can obtain such a form by using the continuum limit of \cref{eq:Pcont}.
Performing a change of variables to $T=e^\tau/N$, we obtain
\begin{equation} 
  p(x,T) = \frac{1}{N}\frac{1}{Tx^2}\exp\left( \frac{1-1/x}{T} \right).
  \label{eq:pxT} 
\end{equation}   
Apart from the normalization $1/N$, this expression is independent of $N$.
The corresponding magnetization $\mu=1/2+\langle S_z\rangle /N$ and radiance $\rho\ce \langle S^+S^-\rangle /N^2$ read (plotted in \cref{fig:convergence}a)
\begin{equation}
  \mu(T)=\frac{e^{1/T}}{T}H(1/T),\quad \rho(T)=\frac 1T-\frac{e^{1/T}}{T^2}(1+T)H(1/T).
  	\label{eq:scaling_form}
\end{equation}
Note that \cref{eq:pxT,eq:scaling_form} are identical to \cref{eq:Pcont,eq:mu_and_rho} up to the change of variables, but we find it more convenient to analyze the behaviour at $N\to\infty$ in this form as they are independent of atom number $N$, having absorbed all dependence on $N$ into $T$. 
Note that the time variable $T(t)=\exp(Nt)/N$ starts at $T(0)=1/N$ and exponentially quickly moves to infinity.
The superradiant pulse occurs close to $T=1$ and, as is evident from \cref{eq:scaling_form}, its height scales with $N^2$. 
We can determine the maximum radiance numerically and find it is $\rho\approx0.196$ at $T\approx1.39$, at which point $\mu\approx 0.532$.
In normal units, the time of maximum radiance is $t_{\mathrm{pulse}}=\ln(N)/N + 0.330/N$.
At $t=0$ we have $T=1/N$, at which point $\mu(1/N)\approx 1-1/N$.
Thus, the starting point of the dynamics corresponds to the time at which on average one photon has decayed.
Another result from the above analysis is that in the original time $t$, the magnetization tends to a step function $\mu(t)\to \Theta(\ln(N)/N-t_{\mathrm{pulse}})$.

\section{Superradiance in the presence of incoherent loss}
Since incoherent loss competes with superradiance and is virtually unavoidable in realistic settings, an important question is whether it may preclude superradiance in certain regimes. 
We answer this question in two parts. First,  we derive the exact solution as $N\to\infty$, which shows that in the large-$N$ limit, superradiance always persists (\cref{sec:loss_solution}).
Second, the question then arises whether there is a critical emitter number that is needed to observe superradiance. This is of experimental relevance particularly in systems in which $\gamma$, the ratio of free-space decay to waveguide decay, is large. We answer this by calculating the threshold emitter number $N$ above which signatures of superradiance emerge as a function of $\gamma$ (\cref{sec:threshold}).
Finally we also compute how many photons decay into free space on average (\cref{sec:incoherent_number}).

\subsection{Asymptotically exact solution}\label{sec:loss_solution}
The asymptotically exact solution $\R(\tau)$ that converges to the exact solution $\PPex$ of \cref{eq:pnr} at all times
is given by
\begin{equation}
  \R_{n,r}(\tau) =
  \begin{cases}
  	\R_{n,r}^<(\tau), &\tau\leq\tau_1,\\
  	\R_{n,r}^>(\tau), &\tau>\tau_1,
  \end{cases}
  \label{eq:Rnr}
\end{equation}
where $\tau_1=(1+\delta_1)\ln N$ with $\delta_1=2/5$, and
\begin{subequations}
  \begin{align}
  	\label{eq:Rnr<}
  	\R_{n,r}^<(\tau)&=R_{n+2r}(\tau)F_{r}(T_n(\tau)),\\
  	\R_{n,r}^>(\tau)&=\sum_{r'=r}^{N/2}
  	\mat{r'\\r}e^{-\gamma r'(\tau-\tau_1)/N}(e^{\gamma (\tau-\tau_1)/N}-1)^{r'-r}\\
  	&\times
  	\sum_{n'=n}^{N-2r'}\mat{n'\\n}e^{-n'(\tau-\tau_1)}(e^{\tau-\tau_1}-1)^{n'-n}\R_{n'r'}(\tau_1),
  	\label{eq:Rnr>}
  \end{align}
  \label{eq:Rnrs}
\end{subequations}
where $F$ is a Poissonian distribution
\begin{equation}
  F_r(\tau) = e^{-\gamma \tau}\frac{(\gamma \tau)^r}{r!},
  \label{eq:F_thm}
\end{equation}
and
\begin{equation}
  T_n(\tau)=\tau+\frac{n}{N}-\frac{n}{n+e^{-\tau}(N-n)}+\ln\left[ \frac{n+e^{-\tau}(N-n)}{N} \right].
  \label{eq:Tn_thm}
\end{equation}

We are interested in the one-norm of the difference of $\R$ and the exact solution $\PPex$ and again allow for arbitrary sequences of times.
\begin{theorem}[Superradiance with incoherent loss]\label{th:loss}
  For any sequence of times $\{\tau_n > 0: n \in \mathbb{N}\}$, 
  \begin{equation}
  	\lim_{N\to \infty} ||\PPex(\tau_N) - \R(\tau_N) ||_1 = 0.
  	\label{eq:convergence_loss}
  \end{equation}
\end{theorem}
The proof closely follows the structure of \cref{th:pure} and can be found in \cref{app:proof2}.
To unpack this result, we highlight a few important properties of superradiance in the presence of incoherent loss below.

\subsection{Qualitative behaviour at large $N$ and threshold}\label{sec:threshold}
At large atom numbers $N$, the average magnetization and radiance are given by \cref{eq:mu_and_rho} up to subleading errors.
While correct, this description misses the fact some excitations are transferred into states that do not decay collectively.
The population of such dark (or subradiant) excitations grows linearly at a rate $\gamma $ until the superradiant burst, such that on average, about $\gamma \ln N$ such dark excitations are produced.
Since they can only decay incoherently, they dominate the late time behaviour with a slow decay at a rate $\gamma/N$.

The result that magnetization and radiance behave as in pure superradiance does 
not necessarily apply to finite $N$ and in fact there is a threshold $N$ below which superradiance disappears.
One necessary requirement for superradiance is that the maximum radiance occurs at $t>0$.
We thus can calculate the threshold above which this is the case directly from the quantum master equation~\eqref{eq:QME} by evaluating 
$\left.\frac{d^2\langle \mu\rangle}{dt^2}\right|_{t=0} = \gamma^2/2+\gamma-N/2+1$,
which implies
$N_{\mathrm{threshold}} = \gamma^2+2\gamma+2$.

\subsection{Number of incoherently decayed photons}\label{sec:incoherent_number}
Since it is clear that superradiance persists in the limit $N\to\infty$, we compute the number of photons lost incoherently during the whole process, $N_{\mathrm{loss}}$, based on the assumption that the system undergoes pure superradiant decay, which becomes exact as $N\to\infty$.
Since the incoherent loss occurs at a rate equal to $\gamma $ times the number of excitations in the system, we integrate the solution \cref{eq:scaling_form} to obtain
\begin{equation}
  N_{\mathrm{loss}} = \gamma \left[\gamma_E + \ln N + e^N H(N)\right]\simeq \gamma (\gamma_E+\ln N),
  \label{eq:lost_photons}
\end{equation}
where $\gamma_E\simeq0.577$ is Euler's constant, and the approximation becomes very good for $N>10$.
Note that \cref{lm:mag} also applies to \cref{eq:lost_photons} as it is just the integral over $\mu(\tau)$.
The leading-order behaviour $N_{\mathrm{loss}}=\gamma \ln N$, valid when $\ln N\gg1\gg N_{\mathrm{loss}}/N$, can be obtained from the simple observation that the superradiant burst occurs at $t\simeq\ln N/N$ and that the emitters are mostly excited before and mostly in the ground state after.
The fact that $N_{\mathrm{loss}}\approx\gamma\ln N$ implies that as $N\to\infty$, the fraction of atoms that decay incoherently is vanishingly small.
This is ultimately the reason why in the large-$N$ limit, the predictions from pure superradiance for magnetization and radiance also hold in the presence of incoherent loss.

\section{Outlook}
Our work improves previous solutions to Dicke superradiance and puts them on a rigorous footing.
There are a number of other experimentally relevant effects that should be considered in future work, some of which may fundamentally change the behaviour at large $N$.
Typical effects include disorder in the decay rates of atoms into the collective mode, inhomogeneous broadening, finite temperature, and, specifically in the case of waveguide QED, spatial disorder. All of these distinguish the atoms and therefore make a straightforward extension of the theory presented here difficult. While inhomogeneous broadening becomes negligible at large $N$, it is not clear at present how the other contributions would affect our solution.

\begin{acknowledgments}
  We would like to thank Nina Fröhling and Tao Shi for insightful discussions.
  We acknowledge funding from ERC Advanced Grant QUENOCOBA under the EU Horizon 2020 program (Grant Agreement No. 742102), and within the D-A-CH Lead-Agency Agreement through Project No. 414325145 (BEYOND C).
  RT acknowledges Max Planck Harvard research center for quantum optics (MPHQ) postdoctoral fellowship. 
\end{acknowledgments}

\appendix
\section{Proof of \cref{th:pure}}\label{app:proof}
We prove \cref{th:pure} in a series of lemmas following the outline given in the main text (see \cref{app:main_lemmas}). Before, we establish some definitions and auxiliary results.

\subsection{Definitions}

\begin{definition}[Exact solution $\Pex$]\label{def:P}
  We denote the exact solution by $\Pex(\tau)=(\Pex_0(\tau),\cdots,\Pex_N(\tau))$. It obeys the equation
  \begin{equation}
  	\partial_\tau \Pex_n(\tau) = -\gamma_n \Pex_n(\tau) + \gamma_{n+1}\Pex_{n+1}(\tau)
  	\tag{\ref{eq:p_n}, restated}
  \end{equation}
  with $\gamma_n=n(N-n+1)/N$,
  which can equivalently be written as the matrix equation
  \begin{equation}
  	\partial_\tau\Pex(\tau) = \Gamma \Pex(\tau).
  	\label{eq:matrix_P}
  \end{equation}
\end{definition}
\begin{definition}[Early solution $Q$]\label{def:Q}
  We denote the solution at early times by $Q(\tau)=(Q_0(\tau),\cdots,Q_N(\tau))$, with components
  \begin{equation}
  	\tag{\ref{eq:Q}, restated}
  	Q_n(\tau) = e^{-\tau}\left( 1-e^{-\tau} \right)^{N-n}.
  \end{equation}
  $Q$ obeys the equation
  \begin{equation}
  	\dot Q_n(\tau) = -\gamma_n^0 Q_n(\tau) + \gamma_{n+1}^0Q_{n+1}(\tau)
  	\label{eq:q_n}
  \end{equation}
  with $\gamma_n^0=N-n+1$,
  which can equivalently be written as the matrix equation
  \begin{equation}
  	\partial_\tau Q(\tau) = \Gamma_0 Q(\tau).
  	\label{eq:matrix_Q}
  \end{equation}
\end{definition}
\begin{definition}[Continuum solution $\Pcont$]\label{def:Rtilde}
  We denote the continuum solution by $\Pcont(\tau)$ with components
  \begin{equation}
  	\Pcont_n(\tau)\ce\frac{N^2}{n^2}\exp\left[ -\tau-e^{-\tau}N\left(\frac{N}{n}-1\right) \right],
  	\tag{\ref{eq:Pcont}, restated}
  \end{equation}
  which, for continuous $n\in(0,N)$ obeys
  \begin{equation}
  	\partial_\tau \Pcont_n(\tau) = \partial_n[n(N-n) \Pcont_n(\tau)]/N.
  	\tag{\ref{eq:pn_continuum}, restated}
  \end{equation}
  We take $\Pcont$ to denote the vector formed by taking integer $n$, $\Pcont=(\Pcont_0, \cdots, \Pcont_N)$.
\end{definition}
Note that \cref{def:Rtilde} is the solution found in the literature.
\begin{definition}[Our solution $R$]\label{def:R}
  We use $R^<=(R^<_0,\cdots, R^<_N)$ to denote our solution
  \begin{equation}
  	\begin{aligned}
  	  R^<_n(\tau) &= \left(\frac{N}{n}\right)^2e^{-\tau}(1-e^{-\tau})^{(N-n)N/n}, \quad R^<_0(\tau) =0,
  	\end{aligned}
  	\tag{\ref{eq:Rn}, restated}
  \end{equation}
\end{definition}

A choice of parameters consistent with all constraints is
\begin{subequations}
  \begin{align}
  	\delta_0 &= 3/4,\qquad \tau_0=\delta_0\ln N\\
  	\delta_1 &= 2/5,\qquad \tau_1=(1+\delta_1)\ln N\\
  	\mu &= 4/5.
  \end{align}
  \label{eq:parameters}
\end{subequations}

\subsection{Additional lemmas}
To bound the differences between the various probability vectors, we need a few additional results.
In the following, $| |\cdot| |$ always denotes the 1-norm.

First we show that in the limit $N\to\infty$, the probability mass of the distributions $Q(\tau)$ and $R(\tau)$ in the interval from $n=0$ to $N-N^\mu$ vanishes faster than any polynomial of $N$. We will use this on many occasions to restrict the range of sums over $n$.
\begin{lemma}[Vanishing probability mass]\label{lm:vanishing_prob_mass}
  If $\tau<\tau_0=\delta_0\ln N$ and $\mu>\delta_0$,
  \begin{equation}
  	\lim_{N\to\infty}\sum_{n=0}^{N-N^\mu}N^kQ_n(\tau)
  	=\lim_{N\to\infty}\sum_{n=0}^{N-N^\mu}N^kR_n(\tau)
  	=0.
  \end{equation}
\end{lemma}
\begin{proof}
  \begin{equation*}
  	\begin{aligned}
  	  \sum_{n=0}^{N-N^\mu}N^kQ_n(\tau)
  	  &< N^k(N-N^\mu)\max_{\tau\in[0,\tau_0]}\max_{n\in[0,N-N^\mu]}e^{-\tau}(1-e^{-\tau})^{N-n}\\
  	  &<N^{1+k}\left( e^{-N^{-\delta_0}} \right)^{N^\mu}\to0.
  	\end{aligned}
  \end{equation*}
  To get to the second line, we use that the maximum is obtained at $n=N-N^\mu$ and we bound $e^{-\tau}\leq1$ and $1-e^{-\tau}\leq1-e^{-\tau_0}$. Similarly,
  \begin{equation*}
  	\begin{aligned}
  	  \sum_{n=0}^{N-N^\mu}N^kR_n(\tau)
  	  &< (N-N^\mu)N^{2+k}\max_{\tau\in[0,\tau_0]}\max_{n\in[0,N-N^\mu]}e^{-\tau}(1-e^{-\tau})^{N(N/n-1)}\\
  	  &<N^{3+k}\left( e^{-N^{-\delta_0}} \right)^{N^\mu}\to0.
  	\end{aligned}
  \end{equation*}
\end{proof}

\begin{lemma}[Normalization of $\Pcont$]\label{lm:Pcont_norm}
  For constant times $\tau$, the normalization of $\Pcont$ obeys
  \begin{equation}
  	\lim_{N\to\infty}| |\Pcont(\tau)| |_1=\frac{ \exp(-\tau-\lambda e^{-\tau})}{1-\exp(-e^{-\tau})}\neq1.
  	\label{eq:pcont_norm_limit}
  \end{equation}
\end{lemma}
\begin{proof}
  We would like to evaluate $| |\Pcont| |_1=\sum_{n=0}^N\Pcont_n$ for constant $\tau$. First note that if $n\leq N-N^\mu$ for any $\mu>0$, we have $\Pcont_n<(N^2/n^2)\exp[-e^{-\tau}{N^\mu}]\to0$ and thus we can restrict attention to $n>N-N^\mu$.
  We define $s=N-n$ and evaluate
  \begin{equation}
  	\begin{aligned}
  	  &| |\Pcont| |_1=\sum_{s=0}^{N^\mu}\left( 1-\frac{s}{N} \right)^{-2}a^{N\left( (N+\lambda)/(N-s)-1\right)}\\
  	  &=e^{-\tau}a^\lambda(1+\OO(N^{\mu-1}))\sum_{s=0}^{N^\mu} a^s a^{\OO(s^2/N)},
  	\end{aligned}
  \end{equation}
  where $a=\exp(-e^{-\tau})$ is a constant with $0<a<1$.
  In the limit $N\to\infty$ the geometric series gives $e^{-\tau}/(1-a)$, which establishes the result. 
\end{proof}

We frequently need to bound sums over $\Pcont$ for times in the range $\tau\in[\tau_0,\tau_1]$, $\tau_1=(1+\delta_1)\ln N$, for which we use the following definition.
\begin{definition}\label{def:G}
  \begin{equation}
  	G_{a,b,c}\ce N\max_{\tau\in[\tau_0,\tau_1]}\max_n\Pcont_n(\tau)n^{-a}N^be^{-c\tau}.
  	\label{eq:Gabc}
  \end{equation}
\end{definition}

\begin{lemma}[Bounds on $G$]\label{lm:G}
  \begin{equation}
  	\begin{aligned}
  	  G_{a,b,c}(\tau)&=
  	  \begin{cases}
  	  	\OO(N^{1+b-a-\delta_0(1+c)}+N^{\delta_1(1+a-c)+b-a-c}) &\mathrm{if}\, 1+a-c>0\\
  	  	\OO(N^{1+b-a-\delta_0(1+c)}+N^{b-c-a})&\mathrm{otherwise},
  	  \end{cases}
  	\end{aligned}
  	\label{eq:Gbound}
  \end{equation}
  where $f(N)=\OO(g(N))$ is standard big-O notation, i.e., $\exists M,N_0>0$ such that $|f(N)|\leq M g(N)$ for all $N\geq N_0$.
\end{lemma}
\begin{proof}
  First, note that we can write
  \begin{equation}
  	G_{a,b,c}=\max_\tau N^{3+b}e^{-(1+c)\tau}\max_n \frac{\exp(-e^{-\tau} N(N/n-1))}{n^{2+a}}.
  \end{equation}
  Let us bound $\tau\leq\ln N$ and $\tau>\ln N$ separately.
  For this we define
  \begin{equation}
  	G_{a,b,c}^<\ce N\max_{\tau\in[\tau_0,\ln N]}\max_n\Pcont_n(\tau)n^{-a}N^be^{-c\tau}.
  	\label{eq:G<}
  \end{equation}
  and
  \begin{equation}
  	G_{a,b,c}^>\ce N\max_{\tau\in[\ln N,\tau_1]}\max_n\Pcont_n(\tau)n^{-a}N^be^{-c\tau}.
  	\label{eq:G>}
  \end{equation}
  For $\tau_1>\tau>\ln N$, the maximum with respect to $n$ is reached at 
  \begin{equation}
  	n_{\mathrm{max}}=\frac{N^2e^{-\tau}}{2+a}
  	\label{eq:maxG_pos}
  \end{equation}
  such that
  \begin{equation}
  	\begin{aligned}
  	  G_{a,b,c}^>&=N^{3+b}\max_{\tau\in[\ln N,\tau_1]}e^{(1+a-c)\tau}e^{Ne^{-\tau}}\left( \frac{2+a}{eN^2} \right)^{(2+a)}\\
  	  &=
  	  \begin{cases}
  	  	\OO(N^{(1+\delta_1)(1+a-c)-1+b-2a})\quad &\mathrm{if}\quad 1+a-c>0,\\
  	  	\OO(N^{b-c-a})&\mathrm{else.}
  	  \end{cases}
  	\end{aligned}
  \end{equation}
  When $\tau\leq\ln N$, the maximum is reached at $n_{\mathrm{max}}>N$.
  Since $n$ can be at most equal to $N$, we instead obtain
  \begin{equation}
  	G_{abc}^<=\max_{\tau}N^{1+b-a}e^{-\tau(1+c)}=N^{1+b-a-\delta_0(1+c)}.
  \end{equation}
  Combining these results, we arrive at \cref{eq:Gbound}.
\end{proof}

\subsection{Main proof}\label{app:main_lemmas}
To prove \cref{th:pure}, we prove the statement separately for short, intermediate, and late times.

\begin{lemma}[$Q$ is correct]\label{lm:Q}
  For any $\tau<\delta_0\ln N$, where $0<\delta_0<1$, $Q(\tau)$ converges to the exact solution $\Pex(\tau)$.
  \begin{equation}
  	\Delta_1(\tau) \ce | |Q(\tau)-\Pex(\tau)| |_1\Nto0.
  	\label{eq:delta1}
  \end{equation}
\end{lemma}
\begin{proof}

Using \cref{eq:matrix_Q,eq:matrix_P} and defining 
$\Gamma_1=\Gamma_0-\Gamma$ with rates $\gamma^1_n=(n/N-1)(N-n+1)$, we have
\begin{equation}
  \Pex(\tau) = e^{\Gamma_0 \tau} \Pex(0) - \int_0^\tau d\tau' e^{\Gamma(\tau-\tau')} \Gamma_1 e^{\Gamma_0\tau'} \Pex(0).
\end{equation}
With this expression and \cref{eq:matrix_Q,eq:delta1} we obtain
\begin{equation}
  \Delta_1(\tau) \le \int_0^\tau d\tau' ||e^{\Gamma(\tau-\tau')} \Gamma_1 Q(\tau')||\leq \int_0^\tau d\tau' ||\Gamma_1 Q(\tau')||,
  \label{eq:delta1_expr}
\end{equation}
since $\exp(\Gamma\tau)$ is a stochastic matrix, which does not increase the norm.

To bound $\Delta_1$, we first evaluate the argument in \cref{eq:delta1_expr} (note $\gamma^1_N=\gamma^1_{N+1}=0$)
\begin{equation}
  \begin{aligned}
  	&||\Gamma_1Q(\tau)|| = \sum_{n=0}^{N-1}|\gamma_{n+1}^1Q_{n+1}(\tau)-\gamma_n^1Q_n(\tau)|\\
  	&=\sum_n\left|\frac{e^{-\tau}(1-e^{-\tau})^{N-n}(N-n)(N+1-n-2e^\tau)}{N(e^\tau-1)}\right|.
  \end{aligned}
  \label{eq:argument}
\end{equation}
For a given $\tau$, the term inside the absolute value signs in \cref{eq:argument} changes sign at $n=N+1-2e^\tau$. 
Thus we can split the sum up into two parts, one from $n=0$ to $n=\bar n \ce \lfloor N+1-2e^{\tau}\rfloor$, and the other from $n=\bar n+1$ to $n=N-1$.
$\Delta_1$ obeys the inequality
\begin{equation}
  \Delta_1(\tau)\le \Delta_{a}+\Delta_{b}
\end{equation}
where
\begin{equation}
  \label{Deltax}
  \Delta_x = \tau_0 \; \max_{\tau\le \tau_0}
  \sum_{n\in \mathfrak N_x}|\gamma^1_{n+1} Q_{n+1}(\tau) - \gamma^1_{n} Q_{n}(\tau)|
\end{equation}
with $x=a,b$, and $\mathfrak N_a=\{0,\ldots,\bar n\}$ and $\mathfrak N_b=\{\bar n+1,\ldots, N\}$.
In each part, all terms have the same sign ($+1$ in the first, $-1$ in second), so we can get rid of the magnitude sign such that only the boundary terms survive.

Considering first $\Delta_b$, we have 
\begin{equation}
  \Delta_b=\tau_0\max_{\tau\leq\tau_0}\gamma^1_{\bar n+1}Q_{\bar n+1}(\tau)
  =\tau_0\max_{\tau\leq\tau_0}\frac{2e^\tau(1-e^{-\tau})^{2e^\tau}(2e^\tau-1)}{(e^\tau-1)N}.
  \label{eq:Deltab}
\end{equation}
The right-hand side of \cref{eq:Deltab} is an increasing function of $\tau$, so we can replace $\tau$ by $\tau_0$, which yields
\begin{equation}
  \Delta_b=\frac{\tau_0}{N^{1-\delta_0}}2(1-N^{-\delta_0})^{2N^{\delta_0}}\frac{2N^{\delta_0}-1}{N^{\delta_0}-1}<\frac{\tau_0}{N^{1-\delta_0}}.
  \label{eq:Deltabbound}
\end{equation}
Clearly this vanishes for $N\to\infty$ for $\delta_0<1$ (see \cref{eq:parameters}).

Turning to $\Delta_a$, we have
\begin{equation}
\Delta_a =\tau_0\max_{\tau\leq\tau_0}(\gamma^1_{\bar n+1}Q_{\bar n+1}(\tau)-\gamma_0^1Q_0(\tau))\leq\Delta_b + \tau_0N(1-e^{-\tau_0})^N.
  \label{eq:Deltaa}
\end{equation}
Thus this vanishes as well, and therefore $\Delta_1(\tau)\to0$.

\end{proof}
\begin{lemma}[$R$ converges to $Q$ at short times]\label{lm:QtoR}
  For any $\tau<\ln N$, $R(\tau)$ converges to the exact solution $Q(\tau)$.
  \begin{equation}
  	\Delta_2(\tau) \ce | |Q(\tau)-R(\tau)| |_1\Nto0.
  	\label{eq:delta2}
  \end{equation}
\end{lemma}
\begin{proof}
Using \cref{eq:delta2,def:Q,def:R}, we have
\begin{equation}
  \Delta_2=\sum_n e^{-\tau}(1-e^{-\tau})^{N-n}\left|1-\frac{N^2}{n^2}(1-e^{-\tau})^{(N-n)(N/n-1)}\right|.
  \label{eq:Q-R}
\end{equation}
Taking a specific time $\tau=\bar\tau\ln N$, we can use \cref{lm:vanishing_prob_mass} to restrict the sum over $n$ to the range $n>N-N^{\bar\tau}$.
Let $s=N-n$. Then
\begin{equation}
  \begin{aligned}
  	\Delta_2&=\sum_{s=0}^{N^{\bar\tau}} e^{-\tau}(1 - e^{-\tau})^s \left|1 - \frac{N^2}{(N-s)^2} (1-e^{-\tau})^{\frac{s^2}{N-s}}\right|\\
  	&=\sum_{s=0}^{N^{\bar\tau}} e^{-\tau}(1 - e^{-\tau})^s \left|1 - (1+\OO(s/N))\left[1+\OO\left(\frac{s^2}{N-s}e^{-\tau}\right)\right]\right|\\
  	&=\OO\left[\max_{s\in[0,N^{\tilde\tau}]}\left(\frac sN+\frac{s^2}{N-s}e^{-\tau}\right)\right],
  \end{aligned}
  \label{eq:other_part}
\end{equation}
where in the last line we used that the sum over $e^{-\tau}(1-e^{-\tau})^s$ is bounded 1.
For shorter times, we can arbitrarily restrict the sum to $s<N^{0.1}$ and the bound still works.
Thus, $\Delta_2\to0$ for all $\tau<(1-\eps)\ln N$ for any constant $\eps>0$.

\end{proof}

\begin{lemma}[Equivalence of $R$ and $\Pcont$]\label{lm:Rtilde}
  If $\tau>\tau_0$,
  \begin{equation}
  	\Delta_3(\tau)\ce| |R(\tau)-\Pcont(\tau)| |\Nto0.
  	\label{eq:RtoRtilde}
  \end{equation}
\end{lemma}
\begin{proof}
  Bounding the residue from Taylor's theorem, we have $\exp(-e^{-\tau})=1-e^{-\tau}+\theta_1$, with $|\theta_1|<e^{-2\tau}$, such that
\begin{equation}
  \begin{aligned}
  	\Delta_3(\tau)
  	&=\sum_{n=0}^N\Pcont_n(\tau)
  	\left|\left( 1+\theta_1e^{-\tau} \right)^{N(N/n-1)}-1\right|\\
  	&<\sum_{n=0}^N\Pcont_n(\tau)\OO(\theta_1 e^{-\tau}N(N/n-1)).
  	\end{aligned}
  \label{eq:rtilde}
\end{equation}
We distinguish three cases. (i) If $\tau>(1+\eps)\ln N$ for any $\eps>0$, then $N^2e^{-2\tau}\to0$ and this expression vanishes.
(ii) If $(1-\eps)\ln N<\tau<(1+\eps)\ln$ we can use \cref{lm:G} to bound \cref{eq:rtilde} by $N^{2-3(1-\eps)}$, which vanishes.
(iii) If $\tau<(1-\eps)\ln N$, we split the sum into two parts, one up to $N-N^\mu$ and the other from $N-N^\mu$ to $N$ (for some $0<\mu<1$).
For the first part, we have 
\begin{equation}
  \begin{aligned}
  	\sum_{n=0}^{N-N^\mu}|R_n-\Pcont_n|&<N^{3}\max_{n\leq N^\mu}e^{-2\tau}\alpha_n
  	\exp\left(-\alpha_n\right)\\
  	&<N^{3-2\delta_0}N^{\mu-\delta_0}\exp(-N^{\mu-\delta_0}),
  \end{aligned}
  \label{eq:bound1}
\end{equation}
where in the first line we introduced $\alpha_n=e^{-\tau}N(N/n-1)$ and to go to the second line we first used that the maxmium with $\alpha$ of $\alpha e^{-\alpha}$ is at $\alpha=1$, but that $\alpha\leq N^{\mu-\delta_0}$ to replace $\alpha=N^{\mu-\delta_0}$, where we also used that $e^{-\tau}\leq N^{-\delta_0}$. \Cref{eq:bound1} vanishes for $\mu>\delta_0$.
For the second part, we have 
\begin{equation}
  \begin{aligned}
  	\sum_{s=0}^{N^\mu}(R_{N-s}-\Pcont_{N-s})&<N^{2\mu-3\delta_0},
  \end{aligned}
  \label{eq:bound2}
\end{equation}
which can be made to vanish, too, by taking $\mu<3\delta_0/2$ (consistent with \cref{eq:parameters}).
\end{proof}
\begin{lemma}[$\Pcont$ is asymptotically correct.]\label{lm:continuum_works}
  For times $\tau_0<\tau<\tau_1$, $\Pcont$ converges to $\Pex$
  \begin{equation}
  	\Delta_4(\tau) = | |\Pex(\tau)-\Pcont(\tau)| |\to0.
  	\label{eq:delta3}
  \end{equation}
\end{lemma}
\begin{proof}
We establish $\Delta_4\to0$ by bounding the difference between the continuum solution (\cref{def:Rtilde}) and the discrete exact solution (\cref{def:P}). 
The exact discrete solution fulfils \cref{eq:p_n},
whereas the continuum solution obeys
\begin{equation}
  \partial_\tau\Pcont_n(\tau) = \partial_n\{[n(N-n)/N]\Pcont_n(\tau)\}.
  \tag{\ref{eq:pn_continuum}, restated}
\end{equation}
We replace the differential by a first-order finite difference with error $U$
\begin{equation}
  \partial_n(\gamma_n\Pcont_n(\tau))=\gamma_{n+1}\Pcont_{n+1}(\tau)-\gamma_n\Pcont_n(\tau)+U_n(\tau),
  \label{eq:replacing_derivative}
\end{equation}
By the remainder theorem applied to the variable $n$, the residue is bounded by
\begin{equation}
  |U_n(\tau)| < \max_{n\leq m\leq n+1}\left( \frac{|H_2(m,\tau)|}{2} \right).
  \label{eq:taylor_error}
\end{equation}
where 
\begin{equation}
  H_2(m,\tau)=\partial_m^2(\gamma_m\Pcont_m(\tau)).
  \label{eq:H2}
\end{equation}
We now bound the error $\Eps=\Pex-\Pcont$.
It is governed by
\begin{equation}
  \begin{aligned}
  	\partial_\tau \Eps_n(\tau) &= -\gamma_n\Pex_n(\tau) + \gamma_{n+1}\Pex_{n+1}(\tau) -\partial_n\{[\gamma_n-n/N]\Pcont_n(\tau)\}\\
  	&=-\gamma_n\Eps_n(\tau)+\gamma_{n+1}\Eps_{n+1}(\tau)+U_n(\tau)+\partial_n[n \Pcont_n(\tau)]/N\\
  	&=-\gamma_n\Eps_n(\tau)+\gamma_{n+1}\Eps_{n+1}(\tau)+\Eps_n^{\mathrm{in}}(\tau).
  \end{aligned}
  \label{eq:difference_evolution}
\end{equation}
Note that in the last line we defined 
\begin{equation}
  \Eps_{n}^{\mathrm{in}}(\tau) \ce U_n(\tau)+\Pcont_n(\tau)/N + n\partial_n\Pcont_n(\tau)/N,
  \label{eq:Epsin}
\end{equation}
with $\Eps_0^{\mathrm{in}}(\tau)=U_0(\tau)=\gamma_1\Pcont_1(\tau)$.
This yields a bound on $\Eps$
\begin{equation}
  \begin{aligned}
  	||\vec\Eps(\tau)|| &\leq \int_{\tau_0}^\tau d\tau' ||e^{\Gamma(\tau-\tau')}\vec\Eps_{\mathrm{in}}(\tau')||
  	\leq\int_{\tau_0}^\tau d\tau'||\vec\Eps_{\mathrm{in}}(\tau')||\\
  	&<(\tau_1-\tau_0)\max_{\tau_0\leq\tau\leq\tau_1}| |\vec\Eps_{\mathrm{in}}(\tau)| |.
  \end{aligned}
  \label{eq:error_bound_Eps}
\end{equation}

We bound the terms in $\Eps^{\mathrm{in}}$ individually.
The contribution due to the Taylor residue is
\begin{equation}
  \begin{aligned}
  	\max_{\tau_0\leq\tau\leq\tau_1}| |U(\tau)| | &\leq \max_{\tau}\sum_n\max_{n\leq m\leq n+1}\left( \frac{|H_2(m,\tau)|}{2} \right)\\
  &<N\max_{\tau}\max_n\frac{|H_2(n,\tau)|}{2}.
  \end{aligned}
  \label{eq:taylor_residue}
\end{equation}
Considering first the $U_0=\Pcont_1(\tau)$ term, we have 
\begin{equation}
  \Pcont_1(\tau)=N^2e^{-\tau}(1-e^{-\tau})^{N(N-1)}\to 0,\quad\text{if }\tau<2\ln N,
  \label{eq:n0Term}
\end{equation}
so we can take $n\geq1$ in the following.

To bound $H_2$, we first note
\begin{equation}
  \begin{aligned}
  	H_2(n,\tau)=\frac{\Pcont_n(\tau)}{2}\left( \frac{2+\frac2N}{n}+\frac{N^3(N-n+1)}{n^3e^{2\tau}}+\frac{2N(n-2-2N)}{e^{\tau}n^2} \right).
  \end{aligned}
  \label{eq:H2_expression}
\end{equation}
Taking the maximum over $n$ and $\tau$ lets us bound the corresponding contribution to $\Eps_{\mathrm{in}}$ in terms of $G$ (see \cref{def:G})
\begin{equation}
  \max_\tau| |U| |<\max_\tau(G_{1,0,0} + G_{3,4,2} + G_{2,3,1} + G_{1,1,1} + G_{2,2,1}).
  \label{eq:bound_U}
\end{equation}
Using \cref{lm:G}, we find that these terms all vanish as long as $\delta_1<1/2$ and $\delta_0>2/3$ (consistent with \cref{eq:parameters}).

We also have to consider the second part of $\Eps_{\mathrm{in}}$ that stems from the error introduced in the rate equation. In terms of $G$, we have
\begin{equation}
  \begin{aligned}
  	&\int_{\tau_0}^{\tau_1}| |\partial_n[\frac{n}{N}\Pcont_n(\tau)]| |<\ln(N)(G_{0,-1,0}+G_{1,0,0}+G_{2,2,1})\to0
  \end{aligned}
  \label{eq:second_part_of_Eps_in}
\end{equation}
again using \cref{lm:G}.
Thus, all terms in $\vec\Eps_{\mathrm{in}}$ vanish, which implies that $| |\vec\Eps| |\to0$ as $N\to\infty$.
\end{proof}

\begin{lemma}[$R^>$ is asymptotically correct.]\label{lm:late_times}
  For times $\tau>\tau_1$, $R^>$ converges to $\Pex$
  \begin{equation}
  	\Delta_5(\tau) = | |\Pex(\tau)-R^>(\tau)| |\Nto0.
  	\label{eq:delta4}
  \end{equation}
\end{lemma}
\begin{proof}
  From \cref{lm:Rtilde,lm:continuum_works}, we know that $R_n^<(\tau_1)$ is a good approximation with the error vanishing, where $\tau_1=(1+\delta_1)\ln N$.
Applying now the linearized equations of motion, we obtain the solution (same as \cref{eq:R>})
\begin{equation}
  \label{Stau}
  R^>(\tau) = e^{\Gamma_S(\tau-\tau_1)}R^<(\tau_1),
\end{equation}
where $\Gamma_S$ is the same matrix as $\Gamma$ but with the replacement $\gamma_n \to \gamma_n^S=n$.
The difference between those two matrices, $\Gamma_1=\Gamma_S-\Gamma$, has rates $\gamma^1_n = n(n-1)/N$. 
In terms of $\Gamma_1$, we have
\begin{equation}
  e^{\Gamma (\tau-\tau_1)}R^<(\tau_1)=e^{\Gamma_S (\tau-\tau_1)} R^<(\tau_1) - \int_{\tau_1}^\tau d\tau' e^{\Gamma(\tau-\tau')} \Gamma_1 e^{\Gamma_S\tau'} R^<(\tau_1).
\end{equation}
Using (\ref{Stau}) and the definition, we have 
\begin{equation}
  \begin{aligned}
  	\Delta_5(\tau) &\le \int_{\tau_1}^\tau d\tau' ||e^{\Gamma(\tau-\tau')} \Gamma_1 R^>(\tau')||
  	= \int_{\tau_1}^\tau d\tau' ||\Gamma_1 R^>(\tau')||
  \end{aligned}
  \label{eq:late_error}
\end{equation}
since $\exp(\Gamma\tau)$ is a stochastic matrix, and thus it does not change the norm.

To bound the integral, we first evaluate the argument
\begin{equation}
  \begin{aligned}
  	&||\Gamma_1R^>(\tau)|| = \sum_{n=0}^{N-1}\frac{n}{N}\left|(n+1)R^>_{n+1}(\tau)-(n-1)R^>_n(\tau)\right|\\
  	&\qquad+(N-1)R^>_N(\tau)\\
  	&=\sum_{n=0}^{N-1}
  \left|\frac{n}{N}\sum_{m=n}^{N-1}\mat{m\\n}(e^{\tau-\tau_1}-1)^{m-n}e^{-m(\tau-\tau_1)}\right.\\
  &\qquad\times\left[ (m+1)R_{m+1}^<(\tau_1)-(n-1)R_m^<(\tau_1) \right] - g_n\Big|\\
  &\qquad+N^{-\delta_1}e^{-N(\tau-\tau_1)},
  \end{aligned}
  \label{eq:late_argument}
\end{equation}
where
\begin{equation}
  g_n=\mat{N\\n}\frac nN(n-1)e^{-N(\tau-\tau_1)}(e^{\tau-\tau_1}-1)^{N-n}R_N(\tau_1).
  \label{eq:gn}
\end{equation}
First, we notice that the term in square brackets in \cref{eq:late_argument} is always positive, as we can bound it from below by 
\begin{equation}
  \begin{aligned}
  	&(m+1)R_{m+1}^<(\tau_1)-(n-1)R_m^<(\tau_1)\\
  	&\geq2R_{m+1}(\tau_1)+(n-1)(R_{m+1}(\tau_1)-R_m(\tau_1))\\
  	&=R_{m+1}(\tau_1)\left\{2+ (n-1)\left[1-\left(1+\frac 1m\right)^2\exp\left( -\frac{N^2e^{-\tau_1}}{m(m+1)}\right) \right] \right\}\\
  	&\geq R_{m+1}(\tau_1)\left\{2+ (n-1)\left[1-\left(1+\frac 1n\right)^2 \right] \right\}\\
  	&=R_{m+1}(\tau_1)(1/n+1/n^2)>0.
  \end{aligned}
  \label{eq:square_bracket_bound}
\end{equation}
Second, we find
\begin{equation}
  K \ce \sum_n g_n = (N-1)e^{\tau_1-2\tau}< N^{-\delta_1}
  \label{eq:gn_sum}
\end{equation}
Thus we can drop the absolute value signs from \cref{eq:late_argument} (since the difference in square brackets is positive)
while incurring an error of at most $2K$.
The sum over $n$ then reduces to the boundary term $N(N-1)R^>_N(\tau)$.
Thus, we have 
\begin{equation}
  \begin{aligned}
  	| |\Gamma_1 R^>(\tau)| |
  	&\le
  	(N-1)R^>_N(\tau)+2K+N^{-\delta_1}e^{-N(\tau-\tau_1)}.
  \end{aligned}
  \label{eq:bound2409}
\end{equation}
The integral
\begin{equation}
  \begin{aligned}
  	\int_{\tau_1}^\tau &d\tau'| |\Gamma_1R^>(\tau')| |
  	\le \frac{N-1}{N}e^{-\tau_1}(1-e^{-N(\tau-\tau_1)})\\
  	&+\frac{N-1}{2}\left( e^{-\tau_1}-e^{\tau_1-2\tau} \right)
  	+N^{-1-\delta_1}(1-e^{-N(\tau-\tau_1)})\\
  	&<N^{-\delta_1}
  \end{aligned}
  \label{eq:delta4_integral}
\end{equation}
vanishes as $N\to\infty$ for any $\tau>\tau_1$.
\end{proof}

Finally, we prove the corollary that states that the formulae for magnetization and radiance are correct.
\begin{L2}[Magnetization and radiance (restated)]\label{lm:mag_app}
  The predicted magnetization and radiance~\eqref{eq:mu_and_rho} are asymptotically correct for any sequence of times 
  $\{\tau_n > 0: n \in \mathbb{N}\}$
  \begin{subequations}
  	\begin{align}
  	  \lim_{N\to\infty}\mu^{\mathrm{exact}}(\tau_N)-\mu(\tau_N)&\to0,\\
  	  \lim_{N\to\infty}\rho^{\mathrm{exact}}(\tau_N)-\rho(\tau_N)&\to0
  	\end{align}
  \end{subequations}
\end{L2}
\begin{proof}
  At early times ($\tau<\tau_0$), we can evaluate the magnetization explicitly using the distribution $Q$ (\cref{lm:Q})
  \begin{equation}
  	\sum_{s=0}^N(N-s)e^{-\tau}(1-e^{-\tau})^s=N-e^\tau+1-\OO( (1-e^{-\tau})^{N-1} ).
  	\label{eq:explicit_mag}
  \end{equation}
  Expanding $H(z)$ around $z=\infty$, we find $H(z)=e^{-z}(z^{-1}-z^{-2}+2z^{-3}+\cdots)$. Therefore, \cref{eq:mu_and_rho} becomes $\mu(\tau)=1-e^\tau/N+e^{2\tau}/N^2+\cdots$.
  Thus, the error in the magnetization is $\OO(1/N, e^{2\tau}/N^2)$.
  Clearly, the same is true for the radiance.

  For intermediate times, we bound the difference between the integral and the sum, $E^{\mathrm{(mag)}}\ce \int_{0}^Ndn\,n\Pcont_n(\tau) - \sum_{n=0}^Nn\Pcont_n(\tau)$.
  Using Taylor's theorem to bound the residue of the zeroth order expansion, we obtain
  \begin{equation}
  \frac{|E^{\mathrm{(mag)}}|}{N}\leq \sum_n \Pcont_n(\tau)\left(\frac 1N+\frac{e^{-\tau}N}{n^2}\right)
  	\label{eq:taylor_diff}
  \end{equation}
  This vanishes by \cref{lm:G}.
  The same steps for the radiance yields the same error bound, but with an additional contribution $\sum_n(N/n)e^{-\tau}\Pcont_n$, which also vanishes by \cref{lm:G}.

  Finally, to cover late times, we prove that at time $\tau_1$, both the exact magnetization and $\mu$ \eqref{eq:mu_and_rho} vanish as $N\to\infty$.
  Specifically, we show that the probability mass of the solution for $n>N^\mu$ for $\mu<1-\delta_1$ vanishes, even after multiplying it by $N^k$ with $k$ constant.
  \begin{equation}
  	\sum_{n=N^\mu}^N N^k \Pcont_n(\tau_1)< \sum_{n=N^\mu}^N N^{1+k}\exp(-N^{1-\delta_1-\mu})\to0.
  	\label{eq:prob_mass_3}
  \end{equation}
\end{proof}

\begin{T1}[Superradiant decay from all-inverted state (restated)]
  For any sequence of times $\{\tau_n > 0: n \in \mathbb{N}\}$, 
  \begin{equation}
  	\lim_{N\to \infty} ||\Pex(\tau_N) - R(\tau_N) ||_1 = 0.
  	\tag{\ref{eq:convergence}, restated}
  \end{equation}
\end{T1}
\begin{proof}
  Together with the triangle inequality, \cref{lm:Q,lm:QtoR} establish \cref{eq:convergence} for times $\tau_N<\ln N$. Similarly, \cref{lm:Rtilde,lm:continuum_works} together show that \cref{eq:convergence} holds between $\tau_0$ and $\tau_1$, and \cref{lm:late_times} proves that \cref{eq:convergence} is true for $\tau_N\geq\tau_1$.
\end{proof}

\section{Proof of \cref{th:loss}}
\subsection{Preliminaries}
\begin{definition}[Exact solution $\PPex$]\label{def:Pnr}
  The exact solution is given by the probability vector $\PPex_{n,r}(\tau)$, where $r\in\{0, \cdots, \floor{N/2}\}$ and $n\in\{0,\cdots, \tildeN \}$.
  The probabilities $\PPex_{n,r}$ obey $\PPex_{n,r}(0)=\delta_{n,N}\delta_{r,0}$ and 
  \begin{equation}
  	\partial_\tau \PPex_{n,r} = -\Gamma^{(1)}_{n,r}\PPex_{n,r}+\Gamma^{(2)}_{n+1,r}\PPex_{n+1,r}+\Gamma^{(3)}_{n+2,r-1}\PPex_{n+2,r-1}+\Gamma^{(4)}_{n,r+1}\PPex_{n,r+1}
  	\tag{\ref{eq:pnr}, restated}
  \end{equation}
  where the rates are given through~\cite{Shammah2018}
  \begin{equation}
  	\begin{aligned}
  	  \Gamma^{(1)}_{n,r} &= \frac{1}{N}n(\tildeN -n+1)+\frac{\gamma }{N}(n+r),\\
  	  \Gamma^{(2)}_{n,r} &= \frac{1}{N}n(\tildeN -n+1)+\frac{\gamma }{N}\frac{n(N+2)(\tildeN -n+1)}{(N-2r)(\tildeN +2)},\\
  	  \Gamma^{(3)}_{n,r} &= \frac{\gamma }{N}\frac{n(n-1)(\tildeN +r+1)}{(N-2r)(\tildeN +1)},\\
  	  \Gamma^{(4)}_{n,r} &= \frac{\gamma }{N}\frac{(\tildeN -n+1)(\tildeN -n+2)r}{(\tildeN +2)(\tildeN +1)}.
  	\end{aligned}
  	\tag{\ref{eq:rates_main}, restated}
  \end{equation}
  We again define the probability vector $\PPex=(\PPex_{0,0},\PPex_{1,0},\cdots)$, and the matrix equation
  \begin{equation}
  	\partial_\tau\PPex(\tau) = \Gamma \PPex(\tau).
  	\label{eq:matrix_P_loss}
  \end{equation}
\end{definition}
\begin{definition}[Approximate solution $\PPap$]\label{def:Pnrtilde}
  We define the approximate solution $\PPap$ as the solution to \cref{eq:pnr} with rates 
  \begin{subequations}
  	\begin{align}
  	  \tilde \Gamma^{(1)}_{n,r} &= \frac{1}{N}n(\tildeN -n+1)+\frac{\gamma }{N}n,\\
  	  \tilde \Gamma^{(2)}_{n,r} &= \frac{1}{N}n(\tildeN -n+1)+\frac{\gamma }{N^2}n(N-n),\\
  	  \tilde \Gamma^{(3)}_{n,r} &= \frac{\gamma }{N^2}n^2,\\
  	  \tilde \Gamma^{(4)}_{n,r} &= 0.
  	\end{align}
  	\label{eq:rates_approx}
  \end{subequations}
\end{definition}

\begin{definition}[Distribution of dark excitations $F$]\label{def:F}
  The distribution in $r$, corresponding to the total spin $j=N/2-r$, is given by the Poisson distribution
  \begin{equation}
  	F_r(\tau) = e^{-\gamma \tau}\frac{(\gamma \tau)^r}{r!}
  	\tag{\ref{eq:F_thm}, restated}
  \end{equation}
\end{definition}
We will frequently need to use that $\langle r\rangle/N$ vanishes, for which we use the following result.
\begin{lemma}[Vanishing probability mass 2]\label{lm:small_r}
  In the limit $N\to\infty$, the probability mass of the distribution $F(\tau)$ vanishes faster than any polynomial of $N$ in the interval from $r=N^\eps$ to $N$ for any $\eps>0$ and $\tau=\OO(N^\delta)$ if $\delta<\eps$.
  Specifically,
  \begin{equation}
  	\lim_{N\to\infty}\sum_{r=N^\eps}^{N}N^kF_r(\tau)
  	=0.
  \end{equation}
\end{lemma}
\begin{proof}
  Using \cref{eq:F_thm}, we have
  \begin{equation*}
  	\lim_{N\to\infty}\sum_{r=N^\eps}^{N}N^kF_r(\tau)<N^{k+1}N^{(\delta-\eps)N^{\eps}}\to0.
  \end{equation*}
\end{proof}

At early times (before $\tau_0$) we linearize the rates around $s=\tildeN -n\approx0$, which yields
\begin{subequations}
  \begin{align}
  	\Gamma^{(1),\mathrm{lin}}_{n,r} &= \tildeN -n+1+\gamma ,\\
  	\Gamma^{(2),\mathrm{lin}}_{n,r} &= \tildeN -n+1,\\
  	\Gamma^{(3),\mathrm{lin}}_{n,r} &= \gamma ,\\
  	\Gamma^{(4),\mathrm{lin}}_{n,r} &= 0.
  \end{align}
  \label{eq:rates_linear}
\end{subequations}
\begin{definition}[Early solution with loss $\Q$]\label{def:Qnr}
  We define $\Q_{n,r}$ by
  \begin{equation}
  	\begin{aligned}
  	  \Q_{n,r}(\tau) = e^{-\tau}(1-e^{-\tau})^{N-2r-n}e^{-\gamma \tau}\frac{(\gamma \tau)^r}{r!}
  	  = Q_{n+2r}(\tau)F_r(\tau),
  	\end{aligned}
  	\label{eq:Qnr}
  \end{equation}
  with $Q_n$ and $F_r$ defined in \cref{def:Q,def:F}.
  $\Q_{n,r}$ obeys \cref{eq:pnr} with rates in \cref{eq:rates_linear} and can thus equivalently be written as
  \begin{equation}
  	\Q(\tau) = e^{\Gamma_0\tau}\PPex(0),
  	\label{eq:Q_matrix_loss}
  \end{equation}
  where $\Gamma_0$ is the same matrix as $\Gamma$ in \cref{def:Pnr}, but with linearized rates \cref{eq:rates_linear}.
\end{definition}

We solve for the dynamics after the initial phase by moving to a continuum limit in $x=n/N$, while retaining $r$ as a discrete variable. 
Specifically, we solve
\begin{equation}
  \partial_\tau p_{n,r}(\tau) - \partial_n\left[\frac{n(N-n)}{N}p_{n,r}(\tau)\right]=\gamma\frac{n^2}{N^2}\left[p_{n,r-1}(\tau)-p_{n,r}(\tau)\right].
  \label{eq:PDE_loss}
\end{equation}
\begin{lemma}[Continuum solution with loss]\label{lm:continuum_loss}
  \Cref{eq:PDE_loss} with initial data at time $\tau=\tau_0$ being $\Pcont_{n+2r}(\tau_0)F_r(\tau_0)$ is solved by
  \begin{equation}
  	\begin{aligned}
  	  \PPcont_{n,r}(\tau)&= \Pcont_{n+2r}(\tau)F_{r}(\tilde T_n(\tau)),
  	\end{aligned}
  	\label{eq:R_app}
  \end{equation}
  where
  \begin{equation}
  	\begin{aligned}
  	  \tilde T_n(\tau)
  	  &= \tau+\left\{\frac{n}{N}-\frac{n}{n+e^{\tau_0-\tau}(N-n)}\right.\\
  	  &+\left.\ln\left[ \frac{n}{N}+e^{\tau_0-\tau}\left(1-\frac nN\right) \right] \right\}.
  	\end{aligned}
  	\label{eq:Tn_app}
  \end{equation}
\end{lemma}
\begin{proof} Proof is by substitution. \end{proof}
In \cref{th:loss}, we take $\tau_0\to0$ in \cref{eq:Tn_app} to obtain \cref{eq:Tn_thm}. The main reason to do so is for simplicity. As we prove below, it does not affect the solution.
Intuitively, the continuum solution \cref{eq:R_app} can be understood by going back to \cref{eq:PDE_loss} and realizing that the left-hand side is the same as \cref{eq:pn_continuum}. We can think of it as the homogeneous part of the differential equation, which is solved by $\Pcont$. The right-hand side only acts on the distribution in $r$ and therefore does not change the distribution in $n$.
Since the right-hand side increases $r$ at a rate depending on $\gamma n^2/N^2$, the distribution in $r$ is a Poisson distribution (\cref{def:F}) but with a time that depends on the dynamics of $n$, which can in turn be related to the characteristics of the left-hand side of \cref{eq:PDE_loss}.

At late times, we linearize around $n=0$ to obtain the rates
\begin{subequations}
  \begin{align}
  	\Gamma^{(1),\mathrm{late}}_{n,r}&=n+\gamma r/N,\\
  	\Gamma^{(2),\mathrm{late}}_{n,r}&=n,\\
  	\Gamma^{(3),\mathrm{late}}_{n,r}&=0,\\
  	\Gamma^{(4),\mathrm{late}}_{n,r}&=\gamma r/N.
  \end{align}
  \label{eq:loss_linearized_rates}
\end{subequations}
To extend the solution to all times, we need to keep small contributions of order $r/N$ in $\Gamma^{(4)}$, as they are responsible for the decay of dark excitations at late times.
The rates predict that the remaining collective and dark excitations decay independently, but their rates have different orders in $N$. 
In the large-$N$ limit, we thus expect first a rapid decay to $n=0$, and subsequently a slow decay of the dark excitations, which can only decay incoherently, as they are dark with respect to the collective decay.

\begin{definition}[Late solution with loss]\label{lm:late}
  We define the late distribution
  \begin{equation}
  	\begin{aligned}
  	  \R^>_{n,r}(\tau) &= \sum_{r'=r}^{N/2}
  	  \mat{r'\\r}e^{-\gamma r'(\tau-\tau_1)/N}(e^{\gamma (\tau-\tau_1)/N}-1)^{r'-r}\\
  	  &\times
  	  \sum_{n'=n}^{N-2r'}\mat{n'\\n}e^{-n'(\tau-\tau_1)}(e^{\tau-\tau_1}-1)^{n'-n}\R_{n'r'}(\tau_1).
  	\end{aligned}
  	\label{eq:R>tau2}
  \end{equation}
  It obeys $\R^>(\tau_1)=\R^<(\tau_1)$ and obeys the equations of motion \cref{eq:pnr} with linearized rates \cref{eq:loss_linearized_rates}.
\end{definition}

\subsection{Main proof}\label{app:proof2}
The proof of \cref{th:loss} has the same structure as the one of \cref{th:pure}. One small difference is that for times shorter than $\tau_1$, we bound the difference to the approximate solution $\PPap$ instead of the exact solution $\PPex$, as it is simpler. \Cref{lm:approximation_converges} shows that this is justified, because the difference between the approximate and exact solution also vanishes for times that grow only logarithmically.

\begin{T2}[Superradiance with incoherent loss (restated)]
  The difference between the exact solution $\rho(\tau)$ to \cref{eq:QME}, subject to the initial condition $\rho(0)=\ket{N}\bra N$,
  and $\sigma^{(N)}(\tau)=\R_{n,r}(\tau)\rhonr_{n,r}$ obeys
  for any sequence of times $\{\tau_n > 0: n \in \mathbb{N}\}$, 
  \begin{equation}
  	\lim_{N\to\infty}||\sigma^{(N)}(\tau_N)-\rho(\tau_N)||_1\to0,\qquad\forall \tau\geq0.
  	\tag{\ref{eq:convergence_loss}, restated}
  \end{equation}
\end{T2}
\begin{proof}
  Applying the triangle inequality twice, \cref{lm:Q_loss,lm:QtoR_loss,lm:approximation_converges} together establish \cref{eq:convergence_loss} for times $\tau_N<\ln N$. Similarly, \cref{lm:Rtilde_loss,lm:continuum_works_loss,lm:approximation_converges} show that \cref{eq:convergence_loss} holds between $\tau_0$ and $\tau_1$, and \cref{lm:late_times_loss} proves that \cref{eq:convergence_loss} is true for $\tau_N\geq\tau_1$.
\end{proof}

\begin{lemma}[$\Q$ is correct]\label{lm:Q_loss}
  For any sequence of times $\{\tau_N < \ln N: N \in \mathbb{N}\}$, 
  $\Q(\tau_N)$ converges to the approximate solution $\PPap(\tau)$.
  \begin{equation}
  	\Delta_1(\tau) \ce | |\Q(\tau)-\PPap(\tau)| |_1\Nto0.
  	\label{eq:delta1_loss}
  \end{equation}
\end{lemma}
\begin{proof}
  As in \cref{lm:Q}, we define the difference between the matrices that generate the evolution $\Gamma_1=\Gamma_0-\tilde\Gamma$, with $\Gamma_0$ defined in \cref{def:Qnr} and $\tilde\Gamma$ defined in \cref{def:Pnrtilde}.
  This allows us to bound $\Delta_1$ as \cref{eq:delta1_expr}
  \begin{equation}
  	\begin{aligned}
  	  \Delta_1&\leq\ln N\max_{\tau\in[0,\tau_0]}| |\Gamma_1\Q(\tau)| |\\
  	&=\ln N\sum_{n,r}\left|-\gamma_{1,nr}^{(1)}\Q_{n,r}(\tau)+\gamma_{1,n+1,r}^{(2)}\Q_{n+1,r}(\tau)\right.\\
  	&\left.\qquad+\gamma_{1,n+2,r-1}^{(3)}\Q_{n+2,r-1}(\tau)\right|,
  	\end{aligned}
  	\label{eq:delta1_bound_loss}
  \end{equation}
  where
  \begin{subequations}
  	\begin{align}
  	  \gamma_{1,nr}^{(1)}&=(\tildeN -n+1)(1-n/N)+\gamma (1-n/N),\\
  	  \gamma_{1,nr}^{(2)}&=(\tildeN -n+1)(1-n/N)-\gamma n(1-n/N)/N,\\
  	  \gamma_{1,nr}^{(3)}&=\gamma (1-n^2/N^2).
  	\end{align}
  	\label{eq:gamma1_rates}
  \end{subequations}
  We split the bound into a part independent of $\gamma $ and one part proportional to $\gamma $, which we will treat separately.
  The independent part is
  \begin{equation}
  	\begin{aligned}
  	  \Delta^{\mathrm{(wg)}}_1 &= \frac{\ln N}{N}\sum_{n,r}\left|(\tildeN -n+1)(N-n)Q_n(\tau)\right.\\
  	  &\left.\qquad-(\tildeN -n)(N-n-1)Q_{n+1}(\tau)\right|F_r(\tau).
  	\end{aligned}
  	\label{eq:Deltawg}
  \end{equation}
  Due to \cref{lm:small_r}, we can neglect the contribution from $r$, such that \cref{eq:Deltawg} reduces to \cref{eq:argument} and thus vanishes.
  The part proportional to $\gamma $ reads
  \begin{equation}
  	\Delta_1^{\mathrm{(free)}}=\gamma \ln N\sum_{n,r}\frac{N-n}{N}\left|-\Q_{n,r}-\frac nN\Q_{n+1,r}+\left(1-\frac nN\right)\Q_{n+2,r-1}\right|.
  	\label{eq:delta1_loss_free}
  \end{equation}
  Note that the rates inside the absolute value sign are all equal or smaller than 1. Using \cref{lm:vanishing_prob_mass}, we can restrict the sum over to values between $\tildeN $ and $\tildeN -N^\mu$, which allows us to bound $N-n$ by $N^\mu$ (recall that $r$ is at most $\ln N$ by \cref{lm:small_r}).
  Since $\Q$ is normalized, we find $\Delta_1^{\mathrm{(free)}}<\gamma (\ln N)^2N^{\mu-1}$, which vanishes as $N\to\infty$.

\end{proof}

\begin{lemma}[Equivalence of $\R$ and $\PPcont$]\label{lm:Rtilde_loss}
  For $\tau_0<\tau<\tau_1$, $\R^<$ [\cref{eq:Rnr<}] converges to $\PPcont$ [\cref{lm:continuum_loss}]
  \begin{equation}
  	\Delta_3(\tau)\ce| |\R(\tau)-\PPcont(\tau)| |\Nto0.
  	\label{eq:RtoRtilde_loss}
  \end{equation}
\end{lemma}
\begin{proof}
  \begin{equation}
  	\begin{aligned}
  	  \Delta_3(\tau)&=\sum_{n,r}|(R_{n+2r}(\tau)-\Pcont_{n+2r}(\tau))F_r(T_n(\tau))\\
  	&-\Pcont_{n+2r}(\tau)[F_r(T_n(\tau))-F_r(\tilde T_n(\tau))]|.
  	\end{aligned}
  	\label{eq:rtilde_bound}
  \end{equation}
  The difference in the first line vanishes by \cref{lm:Rtilde,lm:small_r}.
  The second line is only nonzero for $\tau<\tau_0=\delta_0\ln N$, where $\delta_0<1$ as per \cref{eq:parameters}.
  For those times, we can restrict the sum over $n$ to values between $N-N^\mu$ and $N$ (\cref{lm:vanishing_prob_mass}).
  Defining $s\ce N-n$ and using $s<N^\mu$, we have
  \begin{equation}
  	\begin{aligned}
  	  T_n(\tau) &= \tau+1-\frac{s}{N}-\frac{1-s/N}{1-s(1-e^{-\tau})/N}+\ln\left[ 1-\frac{s}{N}\left( 1-e^{-\tau} \right) \right]\\
  	  &= \tau +\delta\tau_n,\qquad |\delta\tau_n|<N^{\mu-1}.
  	\end{aligned}
  	\label{eq:Tn_bound}
  \end{equation}
  Thus,
  \begin{equation}
  	\begin{aligned}
  	  &\sum_{n,r}\Pcont_{n+2r}(\tau)[F_r(T_n(\tau))-F_r(\tilde T_n(\tau))]|\\
  	  &<\sum_r\sum_{s=0}^{N^\mu}\Pcont_{N-s}(\tau)\left|F_r(\tau)-F_r(\tau+N^{\mu-1})\right|\\
  	  &\leq\sum_r\sum_{s=0}^{N^\mu}\PPcont_{N-s,r}(\tau)\frac{rN^{\mu-1}}{\tau}\to0.
  	\end{aligned}
  	\label{eq:Tn_bound2}
  \end{equation}
  This bound does not work for a sequence of times $\{\tau_n > 0: n \in \mathbb{N}\}$, in which the $\tau_N$ decay as $N^{\mu-1}$ or faster.
  However, in this case, we have $F_{r\geq1}(\tau_N)<N^{r(\mu-1)}$, which means we can replace $F_r$ by $\delta_{r,0}$, in which case we can use \cref{lm:Rtilde}.
\end{proof}

\begin{lemma}[$\R$ converges to $\Q$ at short times]\label{lm:QtoR_loss}
  For any $\tau<\ln N$, $\R(\tau)$ [\cref{eq:Rnr<}] converges to $\Q(\tau)$ (\cref{def:Qnr}).
  \begin{equation}
  	\Delta_2(\tau) \ce | |\Q(\tau)-\R(\tau)| |_1\Nto0.
  	\label{eq:delta2_loss}
  \end{equation}
\end{lemma}
\begin{proof}
We use \cref{lm:QtoR} to replace $Q_n$ by $R_n$, and \cref{lm:vanishing_prob_mass,lm:small_r} to restrict the range of the sum. Using the same reasoning as in \cref{eq:Tn_bound2}
\begin{equation}
  \begin{aligned}
  	\Delta_2(\tau)&\to\sum_{s,r}R_{N-s}(\tau)\left|F_r(\tau)-F_r(T_{N-s})\right|\\
  	&\leq\sum_r\sum_{s=0}^{N^\mu}R_{N-s}(\tau)F_r(\tau)\frac{rN^{\mu-1}}{\tau}\to0.
  \end{aligned}
  \label{eq:delta2_vanishes_loss}
\end{equation}
\end{proof}

\begin{lemma}[$\PPcont$ is asymptotically correct.]\label{lm:continuum_works_loss}
  For times $\tau_0<\tau<\tau_1$, $\PPcont$ [\cref{eq:Rnr<}] converges to $\PPap$ (\cref{def:Pnrtilde})
  \begin{equation}
  	\Delta_4(\tau) = | |\PPap(\tau)-\PPcont(\tau)| |\to0.
  	\label{eq:delta3_loss}
  \end{equation}
\end{lemma}
\begin{proof}
  We start by defining the error
  \begin{equation}
  	\Eps_{n,r}=\PPap_{n,r}(\tau)-\PPcont_{n,r}(\tau),
  	\label{eq:Eps_nr}
  \end{equation}
  where $\PPcont$ solves the continuum equation and has been defined in \cref{lm:continuum_loss}.
  We know that the norm $| |\Eps| |$ vanishes at $\tau_0$ as $N\to\infty$, as per \cref{lm:QtoR_loss,lm:Rtilde_loss}.
  To bound it for the entire time until $\tau_1$, we consider its time evolution
  \begin{equation}
  	\begin{aligned}
  	  &\dot\Eps_{n,r}
  	   =-\tilde\Gamma^{(1)}_{n,r}\PPap_{n,r}+\tilde\Gamma^{(2)}_{n+1,r}\PPap_{n+1,r}+\tilde \Gamma^{(3)}_{n+2,r-1}\PPap_{n+2,r-1}\\
  	   &-\frac{1}{N}\partial_n[n(N-n)\PPcont_{n,r}(\tau)]-\gamma \frac{n^2}{N^2}[\PPcont_{n,r-1}(\tau)-\PPcont_{n,r}(\tau)]\\
  	   &=-\tilde\Gamma^{(1)}_{n,r}\Eps_{n,r}+\tilde\Gamma^{(2)}_{n+1,r}\Eps_{n+1,r}+\tilde \Gamma^{(3)}_{n+2,r-1}\Eps_{n+2,r-1}+\Eps^{\mathrm{in}}_{n,r}.
  	\end{aligned}
  	\label{eq:Eps_eom}
  \end{equation}
  \Cref{eq:Eps_eom} again allows us to bound the error by the integral over the one-norm of $\Eps^{\mathrm{in}}$, as in \cref{eq:error_bound_Eps}.
  To establish this, we need to bound terms of the form
  \begin{equation}
  	\begin{aligned}
  	  A_k&=\sum_{n,r}\frac{n^k}{N^k}\left[\PPcont_{n,r}(\tau)-\PPcont_{n-1,r}(\tau)\right].
  	\end{aligned}
  	\label{eq:bound_this}
  \end{equation}
  First note that due to the prefactor we can neglect the contribution from the sum from $n=0$ to $n=N^{1-\eps}$.
  For the sum from $n=N^{1-\eps}$ to $N$ we use $n\geq N^{1-\eps}$ and \cref{lm:small_r} to write
  \begin{equation}
  	\frac{n^k}{N^k}\leq\frac{(n+2r)^k}{N^k}<\frac{n^k}{N^k}\left( 1+\frac{2k (\ln N)^2}{N^{1-\eps}} \right).
  	\label{eq:nk_bounds}
  \end{equation}
  Thus we have 
  \begin{equation}
  	|A_k|<\sum_{n,r}\frac{n^k}{N^k}\PPcont_{n,r}(\tau)\frac{2k(\ln N)^2}{N^{1-\eps}}\to0.
  	\label{eq:AktoZero}
  \end{equation}
  With this in mind, we write
  \begin{equation}
  	\begin{aligned}
  	  &\Eps^{\mathrm{in}}_{n,r}=F_r(\tilde T_n)\left\{ \gamma_{n+1}\Pcont_{n+1+2r}-\gamma_n\Pcont_{n+2r}-\frac{1}{N}\partial_n\left[ n(N-n)\Pcont_{n+2r} \right]\right\}\\
  	  &\qquad+2r\frac{n}{N}(\PPcont_{n-1,r}-\PPcont_{n,r})+\frac{2r}{N}\PPcont_{n+1,r}\\
  	  &+\frac{\gamma n}{N}\left[\PPcont_{n+1,r}-\PPcont_{n,r}+\frac{n^2}{N^2}\left( \PPcont_{n,r}-\PPcont_{n+1,r}+\PPcont_{n+2,r-1}-\PPcont_{n,r-1}\right) \right].
  	\end{aligned}
  	\label{eq:Eps_loss_bound}
  \end{equation}
  The first line in \cref{eq:Eps_loss_bound} is the same as for pure superradiance \cref{eq:difference_evolution} and thus vanishes.
  In the second line, the first term is of the form \cref{eq:bound_this}, but with an $r$ out front. Since by \cref{lm:small_r}, $r$ grows only logarithmically with $N$, this contribution still vanishes. The second term in the second line vanishes also by \cref{lm:small_r}.
  The terms proportional to $\gamma $ are all of the form \cref{eq:bound_this} and vanish, too.

\end{proof}

\begin{lemma}[Approximate solution converges]\label{lm:approximation_converges}
  For any $\tau<\alpha\ln N$ and some $\alpha>0$, $\PPap(\tau)$ (\cref{def:Pnrtilde}) converges to the exact solution $\PPex(\tau)$ (\cref{def:Pnr}) as $N\to\infty$.
  \begin{equation}
  	| |\PPex(\tau)-\PPap(\tau)| |_1\Nto0.
  	\label{eq:Ptilde_converges}
  \end{equation}
\end{lemma}
\begin{proof}
  Similar to before, we define the error $\Eps(\tau)=\PPex(\tau)-\PPap(\tau)$.
  We bound its norm by
  \begin{equation}
  	| | \Eps(\tau)| |<\alpha\ln N\max_{\tau\in[0,\alpha\ln N]}| |\Gamma_1\PPap(\tau)| |,
  	\label{eq:eps_ptilde}
  \end{equation}
  where $\Gamma_1$ is the difference of the evolution matrices for $\PPex$ (\cref{def:Pnr}) and $\PPap$ (\cref{def:Pnrtilde}).
  Notably, all rates in $\Gamma_1$ are $\OO(r/N)$ and thus can be bounded by $N^{\eps-1}$ for any $\eps>0$.
  Thus,
  \begin{equation}
  	| | \Eps(\tau)| |<\alpha\ln(N)N^{\eps-1}\to0.
  	\label{eq:eps_ptilde_bound}
  \end{equation}
\end{proof}

\begin{lemma}[$\R^>$ is asymptotically correct.]\label{lm:late_times_loss}
  For times $\tau_1<\tau$, $\R^>$ [\cref{eq:R>}] converges to $\PPex$ (\cref{def:Pnr}),
  \begin{equation}
  	\Delta_5(\tau) = | |\R^>(\tau)-\PPex(\tau)| |\to0.
  	\label{eq:delta4_loss}
  \end{equation}
\end{lemma}
\begin{proof}
  From the previous sections, we know that $\R_{n,r}(\tau_1)$ is a good approximation with the error vanishing.
  As before, we can bound $\Delta_5$ by
  \begin{equation}
  	\begin{aligned}
  	  \Delta_5(\tau)
  	  \leq\int_{\tau_1}^\tau d\tau' ||\Gamma_1 \R^>(\tau')||,
  	\end{aligned}
  	\label{eq:late_error_loss}
  \end{equation}
  where $\Gamma_1$ is the difference between the linearized evolution matrix (\cref{lm:late}) and the full evolution (\cref{def:Pnr}).

To bound the integral, we first evaluate the argument
\begin{equation}
  \begin{aligned}
  	&||\Gamma_1\R^>(\tau)|| =
  	\sum_{r=0}^{N/2}\sum_{n=0}^{N-2r}
  \left|-\left[ \frac{n}{N}(1-n-2r)+\frac{\gamma n}{N} \right]\R^>_{n,r}(\tau)\right.\\
&+(n+1)\left[ -\frac{n+2r}{N}+\frac{\gamma (N+2)(N-2r-n+1)}{N(N-2r)N-2r+2)} \right]\R^>_{n+1,r}(\tau)\\
&+\frac{\gamma }{N}\frac{(n+2)(n+1)(N-r+2)}{(N-2r+2)(N-2r+3)}\R^>_{n+2,r-1}\\
&\left.+\frac{\gamma }{N}(r+1)\left[ \frac{(N-2r-n-1)(N-2r-n)}{(N-2r)(N-2r-1)}-1 \right]\R^>_{n,r+1}(\tau)\right|.
  \end{aligned}
  \label{eq:late_argument_loss}
\end{equation}
First we note that subleading terms can be neglected, as even bounding them individually leads to a vanishing contribution. In particular, since we can bound terms containing $n/N$ by $N^{-\delta_1}e^{-\tau-\tau_1}$, all such terms can be neglected.
We are left with
\begin{equation}
  \begin{aligned}
  	&||\Gamma_1\R^>(\tau)|| =
  	\sum_{r=0}^{N/2}\sum_{n=0}^{N-2r-1}
  	\left| \frac{n(n+1)}{N}\R^>_{n+1,r}(\tau)-\frac{n^2}{N}\R^>_{n,r}(\tau)\right|,
  \end{aligned}
  \label{eq:late_argument_loss_approx}
\end{equation}
where we have also dropped the vanishing term
\begin{equation}
  \sum_{r=0}^{N/2}\frac{(N-2r)^2}{N}\R^>_{N-2r,r}(\tau)\to0.
  \label{eq:extra_term_that_vanishes}
\end{equation}

Similar to \cref{lm:late_times}, we can write \cref{eq:late_argument_loss_approx} as
\begin{equation}
  \begin{aligned}
  	&||\Gamma_1\R^>(\tau)|| =
  	\sum_{r=0}^{N/2}\sum_{n=0}^{N-2r-1}
  	\frac{n}{N}\left|\sum_{q=r}^{N/2}\sum_{m=n}^{N-2r'-1}
  \mat{q\\r}\mat{m\\n}e^{-\gamma q(\tau-\tau_1)/N}\right.
  	\\
  	&\times e^{-m(\tau-\tau_1)}(e^{\gamma (\tau-\tau_1)/N}-1)^{q-r}
  	(e^{\tau-\tau_1}-1)^{m-n}\\
  	&\times	\left.\left[ (m+1)R_{m+1}(\tau_1)F_q(T_{m+1}(\tau_1))-(n-1)R_m(\tau_1)F_q(T_m(\tau_1)) \right]\right|
  \end{aligned}
  \label{eq:late_argument_loss_approx2}
\end{equation}
The term in square brackets is again always positive, so we drop the absolute value signs.
To show this, consider
\begin{equation}
  \begin{aligned}
  	&(m+1)R_{m+1}F_{q,m+1}-(n-1)R_mF_{q,m}
  	=R_{m+1}F_{q,m+1}\\
  	&\times\left\{ 2+(n-1)\left[ 1-\left(1+\frac{1}{m}\right)\exp\left( -\frac{N^2e^{-\tau_1}}{m(m+1)} \right)\left(\frac{T_m(\tau_1)}{T_{m+1}(\tau_1)}\right)^q \right] \right\}.
  \end{aligned}
  \label{eq:square_brackets_positive}
\end{equation}
In the regime $m<N^\mu$ for some $\mu<\delta_1/2$ the exponential makes the negative term vanish.
In contrast, if $m>N^\mu$, the fraction $T_m(\tau_1)/T_{m+1}(\tau_1)$ is smaller than 1, such that we can bound it by 1 and use the result \cref{eq:square_bracket_bound}.

Dropping the absolute value signs in \cref{eq:late_argument_loss_approx} means that only the boundary terms survive (up to extra terms that vanish even faster)
\begin{equation}
  \begin{aligned}
  	&||\Gamma_1\R^>(\tau)|| =
  	\sum_{r=0}^{N/2}
  	\frac{(N-2r)^2}{N}\R^>_{N-2r,r}(\tau)<N^{-\delta_1}e^{-N(\tau-\tau_1)}.
  \end{aligned}
  \label{eq:late_argument_loss_approx3}
\end{equation}
Integrating this from time $\tau_1$ to $\infty$ yields an upper bound to the error that goes to zero as $N\to\infty$.

\end{proof}
Note that for this to work it is crucial that we include the decay of $r$ in \cref{eq:loss_linearized_rates}. If we had not included this term in the linearized rates, the distribution in $r$ would remain stationary rather than decay and lead to an error growing in time.

\bibliography{library}
\end{document}